\documentclass[11pt,a4paper]{article}%
\usepackage{amssymb,amsmath, amsfonts}
\usepackage{graphicx,graphics}
\usepackage{mathtools}
\usepackage{bbold}
\usepackage[english]{babel}
\usepackage[utf8]{inputenc}
\usepackage{epsfig,url}
\usepackage{bbm,theorem}
\usepackage{a4wide}
\usepackage{color}
\usepackage{enumerate}
\usepackage{calrsfs}
\usepackage[bookmarks=true]{hyperref}
\usepackage{bookmark}
\usepackage{amsmath}
\usepackage{amsfonts}
\usepackage{amssymb}
\usepackage{verbatim}
\usepackage{graphicx}%
\providecommand{\U}[1]{\protect\rule{.1in}{.1in}}
\providecommand{\U}[1]{\protect\rule{.1in}{.1in}}
\DeclareMathAlphabet{\pazocal}{OMS}{zplm}{m}{n}
\newtheorem{theorem}{Theorem}[section]
\newtheorem{definition}[theorem]{Definition}

\newtheorem{lemma}[theorem]{Lemma}
\newtheorem{proposition}[theorem]{Proposition}

{\theorembodyfont{\upshape}
\newtheorem{remark}[theorem]{Remark}

}
\numberwithin{equation}{section}
\numberwithin{theorem}{section}
\newcommand{\qed}{\hfill$\Box$}
\newenvironment{proof}{\begin{trivlist}\item[]{\em Proof:}\/}{\qed\end{trivlist}}
\newenvironment{proofof}[1][Proof]{\noindent \textit{#1.} }{\ \qed}

\newcommand{\R}{{\mathbb R}}
\newcommand{\C}{{\mathbb C\hspace{0.05 ex}}}
\newcommand{\N}{{\mathbb N}}

\newcommand{\rmd}{{\rm d}}

\newcommand{\beq}{\begin{equation}}
\newcommand{\eeq}{\end{equation}}
\newcommand{\beqs}{\begin{eqnarray}}
\newcommand{\eeqs}{\end{eqnarray}}

\newcommand{\norm}[1]{\Vert #1\Vert}
\newcommand{\defset}[2]{ \left\{ #1\,{:}\,
#2\makebox[0cm]{$\displaystyle\phantom{#1}$}\right\} }

\newcommand{\vep}{\varepsilon}
\newcommand{\defem}[1]{{\em #1\/}}

\newcounter{jlisti}

\newcommand{\cf}[1]{{\mathbbm 1}_{\{#1\}}}

\DeclareMathOperator{\diam}{diam}

\begin{document}

\title{Localization in stationary non-equilibrium solutions for multicomponent coagulation systems}
\author{Marina A. Ferreira, Jani Lukkarinen, Alessia Nota, Juan J. L. Vel\'azquez}
\maketitle

\begin{abstract}
We consider the multicomponent Smoluchowski coagulation equation under non-equi\-librium conditions induced either by a source term or via a constant flux constraint.  
We prove that the corresponding stationary non-equilibrium solutions
have a universal localization property.
More precisely, we show that these solutions asymptotically localize into a direction determined by the source or by a flux constraint: the ratio between monomers of a given type to the total number of monomers in the cluster becomes ever closer to a predetermined ratio as the cluster size is increased. 
The assumptions on the coagulation kernel are quite general, with isotropic power law bounds.
The proof relies on a particular measure concentration estimate and on the control of asymptotic scaling of the solutions which is 
allowed by previously derived estimates on the mass current observable of the system.
\end{abstract}


\tableofcontents

\bigskip

\bigskip

\section{Introduction}

Many processes of particle formation in the atmosphere are due to the
aggregation of gas molecules into small molecular clusters which grow by
colliding with additional gas molecules and with each other (cf.~\cite{Vehkam,Vehkamki2012}). The aggregation process of gas
molecules to form larger clusters is usually described using the so-called
General Dynamic Equation (cf.~\cite{Fried}). This model includes in particular
effects like cluster-cluster reaction as well as cluster-gas reaction.  It also allows for
source terms for small clusters, to take into account physically relevant source processes
such as molecules produced in reactions of vegetation to sunlight.

In many situations of 
interest the particle clusters are made of aggregates of
different types of coagulating molecular types, called 
\defem{monomers}. 
Keeping track of the composition of the clusters, instead only of the total number of particles in them, 
can be important for the applications.
For instance, atmospheric coagulation involves chemically distinct
monomer types such as sulfuric acid and ammonia or dimethylamine molecules which are bases (for an example about the resulting effects on the composition of small clusters we refer to the numerical study 
in \cite{Vehkam}). 

An important common feature of these multicomponent models is that, in the absence of source terms and gelation, 
they conserve the total mass of each monomer type.  Each of these conservation laws is then naturally associated with a current observable, corresponding to the \defem{flux} of the monomer type towards arbitrarily large cluster sizes.
The main motivation for the present work is to develop understanding of the asymptotic, large time and cluster size, properties of the solutions of these equations, both with and without source terms.
We consider here only cases which have a source term or a non-zero constant flux of particles from small to large cluster sizes, and call these \defem{non-equilibrium} solutions.  We also ignore the effect of fragmentation, as motivated later.

Following a common approach in such analysis of asymptotic large-scale properties, the first step is to understand what happens to stationary, i.e., time-independent, solutions to the equations.  In a related work \cite{FLNV1}, we have already extended the analysis from the one-component case in \cite{FLNV}, and we have shown that a natural class of coagulation kernels can be easily partitioned into two subclasses depending on whether a stationary solution exists or not.
Here we continue the analysis of stationary solutions in the class of kernels which can have a stationary 
non-equilibrium solution.  Although not necessarily unique, we will prove that these solutions
share a somewhat surprising common feature which we call \defem{asymptotic localization}: for large clusters, the solutions will become completely localized in one direction.   The direction is determined by the source or flux constraint. Roughly summarized, if one is interested in the monomer composition of a typical large cluster, the ratio of the number of monomers in that cluster will be the same as in the total induced flux, with small corrections which vanish with increasing cluster size.  
(More detailed definitions and discussion will be given Section \ref{sec:asymplocintro}.)

At the first sight, it might appear that this localization is a feature produced by the strong perturbation associated with the constant input of monomers into the system.  Instead, we conjecture that
the localization is a fairly universal feature of the related time-dependent solutions, even without  influx of particles.  Indeed, preliminary results on particular example cases support this conjecture, and we expect to report on more general analysis soon in \cite{FLNV2}.

To formulate the problem mathematically, we label the clusters by 
$\alpha=\left(  \alpha_{1},\alpha
_{2},\ldots ,\alpha_{d}\right)  \in\mathbb{N}_{0}^{d}$, where $d$ denotes the number of 
monomer types and $\alpha_j$ the number of monomers of the type $j$ in the corresponding cluster.
The relevant version of the General Dynamic
Equation which includes coagulation, fragmentation, and monomer injection is then
\begin{equation}
\partial_{t}n_{\alpha}=\frac{1}{2}\sum_{\beta<\alpha}K_{\alpha-\beta,\beta
}n_{\alpha-\beta}n_{\beta}-n_{\alpha}\sum_{\beta>0}K_{\alpha,\beta}n_{\beta
}+\sum_{\beta>0}\Gamma_{\alpha+\beta,\alpha}n_{\alpha+\beta}-\frac{1}{2}%
\sum_{\beta<\alpha}\Gamma_{\alpha,\beta}n_{\alpha}+
\sum_{\left\vert\beta\right\vert =1}
s_{\beta}\delta_{\alpha,\beta} \label{B1}%
\end{equation}
where $\left\vert \alpha\right\vert $ denotes the $\ell^{1}$ norm of $\alpha$,
namely
\begin{equation}
\label{eq:l1norm}\left\vert \alpha\right\vert =\sum_{j=1}^{d}\alpha_{j}.
\end{equation}
Moreover, we will assume that in (\ref{B1}) we have $\alpha\neq O=\left(
0,0,\ldots ,0\right)$.

The coefficients $K_{\alpha,\beta}$ describe the coagulation rate between
clusters with compositions $\alpha$ and $\beta$,
and the coefficients $\Gamma_{\alpha,\beta}$
describe the fragmentation rate of 
clusters of composition $\alpha$ into 
two clusters, one with composition $\beta$ and the other with 
$\left(  \alpha-\beta\right)$.
Suppose that $\alpha=\left(  \alpha_{1},\alpha_{2},\ldots ,\alpha_{d}\right)  $
and $\beta=\left(  \beta_{1},\beta_{2},\ldots ,\beta_{d}\right)  .$ We use the
notation $\beta<\alpha$ to indicate that $\beta_{k}\leq\alpha_{k}$ for all
$k=1,2,\ldots ,d$, and in addition $\alpha\neq\beta.$ We denote as $s_{\beta}$ the
source of monomers characterized by the composition $\beta.$ 
In (\ref{B1}), we have only allowed source terms with
$\left\vert \beta\right\vert =1$ which correspond to injecting clusters 
containing only one monomer, although of any of the $d$ possible types.
We will relax this condition in the results to allow for source terms 
$s_{\beta}$ which are 
supported on a finite set of values $\beta$.

The coefficients $K_{\alpha,\beta}$ yield the coagulation rate between
clusters $\alpha$ and $\beta$ to produce clusters $\left(  \alpha
+\beta\right)  .$ The form of these coefficients depends on the specific
mechanism which is responsible for the aggregation of the clusters.
These coefficients have been computed using kinetic models under different
assumptions on the particle sizes and the processes describing the motion of the
clusters. Various choices of the coagulation kernel have been described in the literature, 
see e.g.\ the textbook \cite{Fried}. 

We will continue to work under the same class of coagulation kernels as in our 
previous work on multicomponent coagulation \cite{FLNV1}.   The precise conditions are collected 
in Section \ref{sec:kernelassumptions}.
In particular, 
the results proven here will apply to the so called 
\emph{diffusive coagulation} or \emph{Brownian kernel} 
\begin{equation}
K_{\alpha,\beta}= C \left(  \frac{1}{\left(  V\left(
\alpha\right)  \right)  ^{\frac{1}{3}}}+\frac{1}{\left(  V\left(
\beta\right)  \right)  ^{\frac{1}{3}}}\right)  \left(  \left(  V\left(
\alpha\right)  \right)  ^{\frac{1}{3}}+\left(  V\left(  \beta\right)  \right)
^{\frac{1}{3}}\right)  \label{KernBrow}%
\end{equation}
where $C>0$ is a constant depending on the process producing the diffusion, 
assuming that the volume scales linearly with the number of monomers in the cluster, i.e., if 
\begin{equation}
k_{1}\left\vert \alpha\right\vert \leq V\left(  \alpha\right)  \leq
k_{2}\left\vert \alpha\right\vert \ \ \ \ \text{with}\quad0<k_{1}\leq
k_{2}<\infty\,.\label{MasVolIneq}%
\end{equation}
The inequalities (\ref{MasVolIneq}) hold, for
instance,  if we suppose $V\left(  \alpha\right)  =\sum_{j=1}^{d}\alpha_{j}v_{j}$ where
$v_{j}>0$ for each $j=1,2,\ldots ,d.$

In this paper, we will ignore the fragmentation terms in (\ref{B1}), and set $\Gamma_{\alpha,\beta}=0.$ The rationale for this, as
 discussed in \cite{FLNV}, is that in the situations we are
interested, the formation of larger particles is energetically favourable and therefore only the coagulation of clusters must be taken into account.  This yields the following evolution problem 
\begin{equation}
\partial_{t}n_{\alpha}=\frac{1}{2}\sum_{\beta<\alpha}K_{\alpha-\beta,\beta
}n_{\alpha-\beta}n_{\beta}-n_{\alpha}\sum_{\beta>0}K_{\alpha,\beta}n_{\beta
}+\sum_{\beta}s_{\beta}\delta_{\alpha,\beta}\,.
\label{B2}%
\end{equation}
We are interested in the steady states of (\ref{B2}) as well as in the steady
states of the continuous version of (\ref{B2}) which is given by 
\begin{align}
& \partial_{t}f\left(  x\right)  =\frac{1}{2}\int_{\left\{  0<y<x \right\}
}dy K\left(  x-y,y\right)  f\left(  x-y\right)  f\left(  y\right)
\nonumber \\ & \qquad
-\int_{ \R^d}dy K\left(  x,y\right)  f\left(  x\right)
f\left(  y\right)  +\eta\left(  x\right) \, ,\quad x\in\mathbb{R}%
^{d}\, ,\ x>0 \,,\label{B3}%
\end{align}
where given $x=\left(  x_{1},x_{2},\ldots ,x_{d}\right)  ,$ $y=\left(  y_{1}%
,y_{2},\ldots ,y_{d}\right)  $ we recall the previously 
introduced comparison notation: $x<y$ whenever $x\le y$ componentwise, and $x\ne y$.
In particular, 
\[
\int_{\left\{  0<y<x \right\}  } dy=\int_{0}^{x_{1}}dy_{1}\int
_{0}^{x_{2}}dy_{2}\cdots \int_{0}^{x_{d}}dy_{d}\,.
\]

Most of the mathematical analysis of coagulation equations has been made for
one-component systems, i.e., with $d=1$. On the other hand, there are only a few papers
addressing the problem of the coagulation equations with injection terms like
$\sum_{\left\vert \beta\right\vert =1}s_{\beta}\delta_{\alpha,\beta}$ or
$\eta.$ This issue has been discussed in \cite{FLNV} and we refer to that
paper for additional references on earlier related works.

\subsection{Coagulation kernel assumptions}\label{sec:kernelassumptions}

We will restrict our attention to the class of coagulation kernels satisfying
the following inequalities 
\begin{align}\label{Ineq1}
c_{1}\left(  \left\vert \alpha\right\vert +\left\vert \beta\right\vert
\right)  ^{\gamma}\Phi\left(  \frac{\left\vert \alpha\right\vert }{\left\vert
\alpha\right\vert +\left\vert \beta\right\vert }\right)   &  \leq
K_{\alpha,\beta}\leq c_{2}\left(  \left\vert \alpha\right\vert +\left\vert
\beta\right\vert \right)  ^{\gamma}\Phi\left(  \frac{\left\vert \alpha
\right\vert }{\left\vert \alpha\right\vert +\left\vert \beta\right\vert
}\right)   ,\ \alpha,\beta\in\mathbb{N}_{0}^{d}\setminus\{O\}
\\
c_{1}\left(  \left\vert x\right\vert +\left\vert y\right\vert \right)
^{\gamma}\Phi\left(  \frac{\left\vert x\right\vert }{\left\vert x\right\vert
+\left\vert y\right\vert }\right)   &  \leq K\left(  x,y\right)  \leq
c_{2}\left(  \left\vert x\right\vert +\left\vert y\right\vert \right)
^{\gamma}\Phi\left(  \frac{\left\vert x\right\vert }{\left\vert x\right\vert
+\left\vert y\right\vert }\right)  ,\, x,y\in\mathbb{R}_{+}^{d}\setminus\{O\}
\label{Ineq2}%
\end{align}
where $p$, $\gamma\in \R$ are some given parameters, $0<c_{1}\leq c_{2}<\infty$, and we define
\begin{equation}
\Phi\left(  s\right)  =\frac{1}{s^{p}\left(  1-s\right)  ^{p}%
}\ \ \text{for \ }0<s<1 \,.\label{Ineq3}%
\end{equation}
Note that then $\Phi\left(  s\right)  =\Phi\left(  1-s\right)$ for all $0<s<1$, and thus 
the above bounds are symmetric under the exchanges $\alpha\leftrightarrow \beta$ and 
$x\leftrightarrow y$.
We stress that even though the bound functions are isotropic, i.e., 
invariant under permutation of components, the kernels need not to be that.
It would be interesting to explore whether localization still holds for strongly non-isotropic kernels.

The existence of steady states to the problems (\ref{B2}),\ (\ref{B3})
has been considered in \cite{FLNV1}.
It is proven there that in the
multicomponent case and under the assumptions (\ref{Ineq1})--(\ref{Ineq3}),
there exists a stationary solution to (\ref{B2}), (\ref{B3}) if and only if 
\begin{equation}
\gamma+2p<1\, . \label{ExCond}%
\end{equation}
In particular, since for $\gamma+2p\ge 1$ there are no stationary solutions, we can leave those parameter values out of consideration here.
(Further discussion about the physical significance and qualitative explanation for the non-existence of stationary solutions for $\gamma+2p\ge 1$ can be found in \cite{FLNV1}.)

It is now straightforward to check that for the diffusive kernel in 
(\ref{KernBrow}) and assuming (\ref{MasVolIneq}) is valid, the assumptions are satisfied after choosing
$\gamma=0$ and $p=\frac{1}{3}$.  Thus 
the inequality (\ref{ExCond}) then holds, and there exists at least one stationary solution.

\subsection{Asymptotic localization}\label{sec:asymplocintro}

The main result of this paper is  a property of the steady
states to (\ref{B2}), (\ref{B3}) which is specific to the multicomponent
coagulation system, i.e., occurs only if $d>1.$ This property, called here
\textit{localization}\/,  
consists in the fact that the mass in the stationary
solutions to (\ref{B2}), (\ref{B3}) concentrates for large values of the
cluster size $\left\vert \alpha\right\vert $ or $\left\vert x\right\vert $
along a specific direction of the cone ${\mathbb{R}}^{d}_{+}$. More precisely, we can find  some relative width  of the strip, $\zeta>1$, %
such that if $\left\{  n_{\alpha}\right\}  _{\alpha\in\mathbb{N}_{0}%
^{d}\setminus\left\{  O\right\}  }$ or $f$ are solutions to (\ref{B2}) or (\ref{B3}),
respectively, there exists a vector $\theta\in\mathbb{R}_{+}^{d}$
satisfying $\left\vert \theta\right\vert =1$ such that for any $\varepsilon>0$
the following inequalities hold (respectively for $\left\{  n_{\alpha}\right\}
_{\alpha }$ or $f$) 
\begin{equation}
\lim_{R\rightarrow\infty}\frac{\sum_{\left\{  R\leq\left\vert \alpha
\right\vert \leq \zeta R\right\}  \cap\left\{  \left\vert \frac{\alpha}{\left\vert
\alpha\right\vert }-\theta\right\vert <\varepsilon\right\}  }n_{\alpha}}%
{\sum_{\left\{  R\leq\left\vert \alpha\right\vert \leq\zeta R\right\}  }n_{\alpha}%
}=1\ \ \text{or\ \ }\lim_{R\rightarrow\infty}\frac{\int_{\left\{
R\leq\left\vert x\right\vert \leq\zeta R\right\}  \cap\left\{  \left\vert \frac
{x}{\left\vert x\right\vert }-\theta\right\vert <\varepsilon\right\}
}f\left(  dx\right)  }{\int_{\left\{  R\leq\left\vert x\right\vert
\leq\zeta R\right\}  }f\left(  dx\right)  }=1. \label{LocStat}%
\end{equation}

In fact, the direction $\theta$ can be uniquely determined from the source term,
$s_{\beta}$ or $\eta$, in the sense that $\theta_j$ will agree with 
the total relative injection rate of monomers of type $j$.  Thus 
the flux of monomers towards large cluster sizes occurs via clusters
with essentially fixed relative monomer compositions.
Let us remark that \eqref{LocStat} is a non-equilibrium property which cannot be
derived from a variational principle such as by minimization of the free energy or any other thermodynamic potential.  On the contrary, the localization property emerges as a consequence of the 
coagulation dynamics. We also point out that for systems of the form \eqref{B2},
\eqref{B3} the detailed balance property fails.

Asymptotic localization appears to be a very generic feature of
multicomponent coagulation, including time-dependent problems. 
This has been shown to occur for the constant kernel $K=1$ and the additive kernel $K=x+y$ for the discrete coagulation equation 
in \cite{KbN}, using the fact that for these kernels the
solutions to \eqref{B2} can be obtained explicitly by means of the generating function method. In a forthcoming paper \cite{FLNV2}, a similar 
localization property will be
shown for mass conserving solutions of the coagulation equation
(i.e.\ with $\eta=0$ or $s_{\beta}=0$), asymptotically for long times.

\subsection{Main notations and structure}\label{sec:notation}

In this paper we will denote the non-negative real numbers and integers by $\mathbb{R}_{+}:={[}0,\infty {)} $ and $\N_0:= \{0,1,2,\ldots\}$, respectively.  We also 
use a subindex ``$*$'' to denote restriction of real-component vectors $x$ to those which satisfy $x>0$, i.e., for which $x_i>0$ for some $i$.  In particular, we denote
$\R_*:=\R_+\setminus\{0\}$,
$\R^d_*:=\R_+^d\setminus\{O\}$ and $\N^d_*:=\N_0^d\setminus\{O\}$.

Assuming that $X$ is a locally compact Hausdorff space, for example $X=\R^d_*$,
we denote with
$C_{c}\left(X\right) $ the space of compactly supported
continuous functions from $X$ to $\C$, and let $C_0(X)$ denote its completion in the standard sup-norm. Moreover, we will denote as
$\mathcal{M}_{+}(X)$ the collection of non-negative Radon measures on $X$, not necessarily
bounded, and as $\mathcal{M}_{+,b}(X)$ its subset consisting of bounded
measures.  We recall that $ \mathcal{M}_{+}(X) $ can be identified with the space of positive linear functionals on $C_{c}(X)$ via Riesz--Markov--Kakutani Theorem.

We will use indistinctly $\eta(x) dx$ and $\eta(dx)$ to denote elements of the above measure spaces.
The notation $\eta( dx) $ will be preferred when performing integrations or when we want to emphasize that the measure might not be absolutely continuous with respect to the Lebesgue measure.  In addition, ``$dx$'' will often be dropped from the first notation, typically when the measure eventually turns out to be absolutely continuous.
We will also borrow a convenient notation from physics to denote ``Dirac $\delta$''-measures:
if $X$ is locally compact Hausdorff space and $x_0\in X$, we denote the bounded positive Radon measure
defined by the functional $\Lambda_{x_0}[f]=f(x_0)$, $f\in C_{c}(X)$, by ``$\delta(x-x_0)\rmd x$''.

We also use the notation $\cf{P}$ to denote the generic characteristic function of a condition $P$: $\cf{P}=1$ if the condition $P$ is
true, and otherwise $\cf{P}=0$.

The plan of the paper is the following. In Section \ref{StatDef} we introduce the definitions of  stationary injection solutions (continuous and discrete) and constant flux solutions that are considered in this paper and we recall the existence results which have been obtained in \cite{FLNV1}. Section \ref{sec:4} contains the main results on mass localization  in stationary solutions. The proofs are  presented in Sections \ref{sec:localization} and \ref{sec:loca_constantFlux} and they use some technical results that are collected in Section \ref{sec:8}.

\section{Classes of steady state
solutions\label{StatDef}}

\subsection{Stationary injection solutions}

The stationary solutions of (\ref{B2}), (\ref{B3}) are described respectively
by the equations 
\begin{align}
&0=\frac{1}{2}\sum_{\beta<\alpha}K_{\alpha-\beta,\beta}n_{\alpha-\beta}%
n_{\beta}-n_{\alpha}\sum_{\beta>0}K_{\alpha,\beta}n_{\beta}+s_{\alpha}\ \ ,\ \ \alpha\in\mathbb{N}_{*}^{d} \,,\label{S1E5} \\
& 0=\frac{1}{2} \int_{\left\{  0<y <x \right\}  } K\left(  x-y,y\right)
f\left(  x-y,t\right)  f\left(  y,t\right)  dy-\int_{\mathbb{R}_{*}^{d}}K\left(
x,y\right)  f\left(  x,t\right)  f\left(  y,t\right)  dy
\nonumber \\ & \qquad 
+\eta\left(  x\right)
\,, \quad x\in\mathbb{R}_{*}^{d}\,.
\label{eq:contStat}%
\end{align}
We will assume that the sources  $s_{\beta}$ and $\eta$ are compactly supported.
For examples of how to relax the assumptions about the source, we refer to a recent paper \cite{Laurencotpreprint20} where compact support is not required assuming that the solution $f$ is absolutely continuous with respect to the Lebesgue measure.

\begin{definition}\label{StatInjDef}
Let $\eta\in\mathcal{M}_{+,b}\left(  \mathbb{R}_{*}^{d}\right)  $
with support contained in the set $\defset{x\in \R^d_*}{1\le |x|\le L}$ 
for some $L>1.$ Suppose that $K$ is continuous and it satisfies (\ref{Ineq2}),
(\ref{Ineq3}). We will say that $f\in\mathcal{M}_{+}\left(  \mathbb{R}_{*}%
^{d}\right)  $ is a stationary injection solution to (\ref{eq:contStat}) if
$f$ is supported in $\defset{x\in \R^d_*}{|x|\ge 1}$ and satisfies
\begin{equation*}
\int_{\R^d_*} |x|^{\gamma+p} f(dx)
<\infty\label{S1E4}
\end{equation*}
as well as the identity%
\begin{equation}
0=\frac{1}{2}\int_{{{{\mathbb{R}}}_{*}^{d}}}\int_{{{{\mathbb{R}}}_{*}^{d}}%
}K\left(  x,y\right)  \left[  \varphi\left(  x+y\right)  -\varphi\left(
x\right)  -\varphi\left(  y\right)  \right]  f\left(  dx\right)  f\left(
dy\right)  +\int_{{{{\mathbb{R}}}_{*}^{d}}}\varphi\left(  x\right)
\eta\left(  dx\right)  \label{WeakForm1}%
\end{equation}
for any test function 
$\varphi\in C^{1}_c\left(  {{{\mathbb{R}}}_{*}^{d}}\right)$.
\end{definition}

For the definition, we recall the notation $\R^d_*:=\R^d_{+}\setminus \{O\}$ and that 
we use the $\ell_1$-norm in $\R^d_*$, i.e., $|x|=\sum_j |x_j|$.

\bigskip

In order to define stationary injection solutions for the discrete equation
(\ref{S1E5}) we use the fact that the solutions of (\ref{S1E5}) can be
thought as solutions $f$ of (\ref{eq:contStat}) where $f$ is supported at the
elements of $\mathbb{N}_{*}^{d}=\N_0^d\setminus\{O\}$.
More precisely, suppose that the sequence
$\left\{  n_{\alpha}\right\}  _{\alpha\in\mathbb{N}_{*}^{d}}$ is a solution of
(\ref{S1E5}) 
and  assume that there is $L>1$ such that
$s_\alpha=0$ whenever $|\alpha|>L$.
We define 
\begin{equation}
f\left(  x\right)  =\sum_{\alpha\in\mathbb{N}_{*}^{d}}n_{\alpha}\delta\left(
x-\alpha\right)  \label{S1E2}%
\end{equation}
as well as 
\begin{equation}
\eta\left(  x\right)  =\sum_{\alpha\in\mathbb{N}_{*}^{d}}s_{\alpha}%
\delta\left(  x-\alpha\right) \,. \label{S1E3}%
\end{equation}
Then $\eta$ satisfies 
the assumptions of Definition
\ref{StatInjDef} with the same parameter $L$.

With these identifications in mind, we can 
then define a solution of (\ref{S1E5}) as follows:
\begin{definition}
\label{DefStatInjDis}Suppose that $\left\{  s_{\alpha}\right\}  _{\alpha
\in\mathbb{N}_{*}^{d}  }$ is a non-negative sequence supported
in a finite collection of values $\alpha.$ We say that a sequence $\left\{
n_{\alpha}\right\}  _{\alpha\in\mathbb{N}_{*}^{d}
}$, with $n_{\alpha}\geq0$ for $\alpha\in\mathbb{N}_{*}^{d}$, is a stationary injection solution of (\ref{S1E5}) if
\[
 \sum_{\alpha\in\mathbb{N}_{*}^{d}} |\alpha|^{\gamma+p}n_{\alpha} < \infty\,, 
\]
and
the measure $f\in\mathcal{M}_{+}\left(  \mathbb{R}_{*}^{d}\right)  $ defined
by means of (\ref{S1E2}) 
is a solution of (\ref{eq:contStat}) in the sense of Definition
\ref{StatInjDef}.
\end{definition}
Note that the assumptions made about $n_\alpha$ guarantee that the associated measure
$f$ is supported in $\defset{x\in \R^d_*}{|x|\ge 1}$ and satisfies (\ref{S1E4}).

Next we state the existence of stationary injection solutions to (\ref{eq:contStat}) and (\ref{S1E5}) for kernels $K$ with $\gamma+2p<1.$   This result has been proved in \cite{FLNV1} and it is a natural extension from one- to multi-component systems of the result contained in \cite{FLNV}. 
\begin{theorem}\label{ThmExistCont}
Suppose that $K$ is a continuous symmetric function that satisfies
(\ref{Ineq2}), (\ref{Ineq3}) with $\gamma+2p<1.$ We have the following results:
\begin{itemize}
\item[(i)] Suppose that 
$\eta\in\mathcal{M}_{+,b}\left(  \mathbb{R}_{*}^{d}\right)  $
is supported inside the set $\defset{x\in \R^d_*}{1\le |x|\le L}$ for some $L>1$. 
Then, there exists a stationary injection solution $f\in\mathcal{M}%
_{+}\left(  \mathbb{R}_{*}^{d}\right)  $ to (\ref{eq:contStat}) in the sense
of Definition \ref{StatInjDef}. 
\item[(ii)]\label{ThmExistDisc}
Suppose that $K$ is symmetric and satisfies (\ref{Ineq1}), (\ref{Ineq3})
with $\gamma+2p<1.$ Let $\left\{  s_{\alpha}\right\}  _{\alpha\in
\mathbb{N}_{*}^{d}  }$ be a non-negative sequence
supported on a finite number of values $\alpha.$ Then, there exists a
stationary injection solution $\left\{  n_{\alpha}\right\}  _{\alpha
\in\mathbb{N}_{*}^{d}  }$ to (\ref{S1E5}) in the
sense of Definition \ref{DefStatInjDis}.
\end{itemize}
\end{theorem}
\begin{remark}
Let us point out that no uniqueness of the solutions
is claimed in Theorems \ref{ThmExistCont}, \ref{ThmExistDisc}. The issue of uniqueness of stationary injection
solutions is an interesting open problem. Moreover, we remark that restrictions for the values of $\gamma$ an $p$ in Theorem \ref{ThmExistCont} are not only
sufficient to have stationary injection solutions, but they are also necessary. Indeed, for values of $\gamma$ an $p$ such that $\gamma+2p\geq1$ it is possible to prove that no such solutions exist (cf.~\cite{FLNV1}). 
\end{remark} 

\medskip

\subsection{Change to $(r,\theta)-$variables}\label{sec:change_var}

The definition of constant flux solution presented in the next section, as well as the main results of this paper are more conveniently written in  the coordinates $\left(  r,\theta\right)$
with $r=|x|>0$ and $\theta = \frac{1}{|x|}x\in \Delta^{d-1}$ for $x\in \R^d_*$, where  $\Delta^{d-1}$ denotes the simplex
\begin{equation*}
\label{eq:simplex}
\Delta^{d-1}=\left\{  \theta\in\mathbb{R}_{*}^{d} :\left\vert \theta\right\vert =1 \right\}.
\end{equation*}
We use the variables $\theta_{1},\theta_{2},\ldots ,\theta_{d-1}$ to
parametrize the simplex and set then 
$
\theta_{d}=1-\sum_{j=1}^{d-1}\theta_{j}\,.
$
Thus, the change of variables is $x\rightarrow\left(  r,\theta_{1},\theta
_{2},\ldots ,\theta_{d-1}\right)$. This map is a bijection from $\R^d_*$ to $\R_*\times \Delta^{d-1} $ and the inverse map  $\R_*\times \Delta^{d-1} \to \R^d_* $ is defined by $x=r\theta ,\ r>0,\ \theta\in\Delta^{d-1}$.

We now rewrite equation \eqref{eq:contStat} using this change of variables. For further details we refer to \cite{FLNV1}.
Computing the determinant of the Jacobian of the mapping $\left(  r,\theta\right) \rightarrow x,$ we obtain
\[
dx=\left\vert \frac{\partial\left(  x_{1},x_{2},\ldots ,x_{d}\right)  }%
{\partial\left(  r,\theta_{1},\theta_{2},\ldots ,\theta_{d-1}\right)  }\right\vert
drd\theta_{1}d\theta_{2}\ldots d\theta_{d-1}=r^{d-1}drd\theta_{1}d\theta
_{2}\ldots d\theta_{d-1}.
\]
Rewriting  $d\theta_{1}d\theta_{2}\ldots d\theta_{d-1}$ in terms of the area
element of the simplex, denoted by $d\tau\left(  \theta\right)$, yields
$d\tau\left(  \theta\right)  =\sqrt{d}\, d\theta_{1}d\theta_{2}\ldots d\theta_{d-1}$ and
\begin{equation*}
dx=\frac{r^{d-1}}{\sqrt{d}}drd\tau\left(  \theta\right) \,. \label{Jac}%
\end{equation*}
Here we used the formula $d\tau\left(  \theta\right)  =\sqrt{1+\left(  \nabla_{\theta}h\right)  ^{2}%
}d\theta_{1}d\theta_{2}\ldots d\theta_{d-1}$ with $\theta_{d}=h\left(  \theta_{1},\theta_{2},\ldots ,\theta_{d-1}\right):= 1-\sum_{j=1}^{d-1}\theta_{j}  .$

Suppose that $x=r\theta$ and $y=\rho\sigma.$
We define the kernel function in these new variables by
\begin{equation}
G\left(r,\rho;\theta,\sigma\right)= K\left(r\theta,\rho\sigma\right)\,,  \label{G1}%
\end{equation}%
and replace the measure $f$ by  
the measure $F\in \mathcal{M}_+(\R_*\times \Delta^{d-1})$ which is
uniquely fixed by the requirement that 
\begin{equation}
\int \psi(r,\theta)\frac{r^{d-1}}{\sqrt{d}} F(r,\theta)drd\tau\left(  \theta\right) 
=\int \psi(|x|,x/|x|) f(x) dx\,, \label{G2}
\end{equation}
for all test functions $\psi\in  C_c(\R_*\times \Delta^{d-1})$.
The normalization with the Jacobian is made to guarantee that, if $f$ is absolutely continuous, then
the respective density functions transform as expected.

Then, 
we obtain that (\ref{WeakForm1}) is equivalent to
\begin{align}
&  \frac{1}{2d}\int_{\mathbb{R}_{*}}r^{d-1}dr\int_{\mathbb{R}_{*}}\rho
^{d-1}d\rho\int_{\Delta^{d-1}}d\tau\left(  \theta\right)  \int_{\Delta^{d-1}%
}d\tau\left(  \sigma\right)  G\left(  r,\rho;\theta,\sigma\right)  F\left(
r,\theta\right)  F\left(  \rho,\sigma\right) \nonumber\\
&  \times\left[  \psi\left(  r+\rho,\frac{r}{r+\rho}\theta+\frac{\rho}{r+\rho
}\sigma\right)  -\psi\left(  r,\theta\right)  -\psi\left(  \rho,\sigma\right)
\right]  +\int_{{\mathbb{R}_{*}^{d}}}\varphi\left(  x\right)  \eta\left(
dx\right)  =0\,.\label{weak_rtheta}%
\end{align}
for all $j=1,2,\ldots ,d$ and $\psi\in C_c(\R_*\times \Delta^{d-1})$ and  using $\varphi(x) = \psi(r,\theta)$ and $\varphi(y) = \psi(\rho,\sigma)$. 

Note  that the change of measure in \eqref{G2} from $f$ to $F$ should be understood 
via the Riesz--Markov--Kakutani Theorem applied to the linear functional
\[
\psi \mapsto  \int_{\R^d_*} \frac{\sqrt{d}}{|x|^{d-1}}
\psi(|x|,x/|x|) f(x) dx\,, \qquad \psi \in C_c(\R_*\times \Delta^{d-1})\,.
\]
Clearly, if $F(r,\theta)dr d\tau(\theta)$ denotes the unique non-negative Radon measure associated with this functional, then 
$F$ satifies also (\ref{G2}).  Note also that if $f$ satisfies the assumptions in 
Definition \ref{StatInjDef}, we also have that the support of $F$ 
lies in $[1,\infty)\times \Delta^{d-1}$ and $F$ satisfies
\[
 \int_{\R_*\times \Delta^{d-1}} r^{d-1+\gamma+p} F(r,\theta)dr d\tau(\theta)
<\infty \,.
\]

\subsection{Constant flux solutions}

Finally, let us detail also the definition of the constant flux solutions associated to
the equation (\ref{eq:contStat}) with $\eta=0$, i.e., to the equation
\begin{equation}
0=\frac{1}{2}\int_{\left\{ 0<y<x\right\} }K\left(  x-y,y\right)  f\left(
x-y,t\right)  f\left(  y,t\right)  dy-\int_{\mathbb{R}_{*}^{d} }K\left(
x,y\right)  f\left(  x,t\right)  f\left(  y,t\right)  dy\,. \label{S1E7}%
\end{equation}

\begin{definition}\label{DefConstFlux}
Suppose that $K$ is continuous and it satisfies
(\ref{Ineq2}), (\ref{Ineq3}). 
We say that $f\in\mathcal{M}\left(
\mathbb{R}_{*}^{d}\right)  $ is a stationary solution of (\ref{S1E7}) 
if 
\begin{equation*}
\int_{\defset{x\in \R^d_*}{|x|\ge 1}} |x|^{\gamma+p} f(dx)
+
\int_{\defset{x\in \R^d_*}{|x| < 1}} |x|^{1-p} f(dx)
<\infty\label{cond_constFlux}
\end{equation*}
 is satisfied and
(\ref{WeakForm1}) holds with $\eta=0$,  
for every test function 
$\varphi\in C^{1}_c\left(  {{{\mathbb{R}}}_{*}^{d}}\right)$.

We then define the \defem{total flux across
the surface $\left\{  \left\vert x\right\vert =R\right\}  $}
as the vector-valued function $A(R)\in \R_+^d$, $R>0$, defined by
\begin{equation}
A_j(R) =  \frac{1}{d}
\int_{(0,R]\times \Delta^{d-1}}dr d\tau\left(  \theta\right)\int_{(R-r,\infty)\times \Delta^{d-1}} d\rho d\tau\left(  \sigma\right)
r^{d}\rho^{d-1}
F\left(  r,\theta\right)  F\left(  \rho,\sigma\right)  \theta_j G\left(
r,\rho;\theta,\sigma\right) \,,\label{S1E8}
\end{equation}
where the function $G$ is defined as in (\ref{G1})
and the measure $F$ using (\ref{G2}), as explained above.
We say that $f$ is a \defem{non-trivial constant flux solution} of (\ref{S1E7}) if it is a stationary solution and there is  
$J_0>O$ 
such that $A(R)=J_0$ for all  $R>0$.
\end{definition}

\begin{remark} 
The flux (\ref{S1E8}) is obtained by considering in \eqref{weak_rtheta} the test function $\psi(r,\theta)=r\chi_\delta(r)$, with $\chi_\delta(r) \in C_c^\infty(\R_*)$ such that $\chi_\delta(r) \in [0,1]$, $\chi_\delta(r)=1$ for $r \in [1,z]$ and $\chi_\delta(r)=0$ for $r \geq z+\delta$ and computing the limit when $\delta\to 0$ following similar arguments as in the proof of Lemma 2.7 in \cite{FLNV}.
 We refer to \cite{FLNV1} for the details of the computations.
Note that in the one-component equation the fluxes being constant implies that $f$ is a solution to the coagulation equation. The same is not true in the multicomponent equation, which justifies the need to assume \eqref{WeakForm1}.
\end{remark}

\begin{remark}
Assume that the kernel $K$ is continuous and homogeneous with homogeneity $\gamma$. If $K$  satisfies (\ref{Ineq2}), (\ref{Ineq3}) with $\gamma+2p<1$, then one can show (cf. \cite{FLNV1}) that a family of constant flux solutions of (\ref{S1E7}) is given by the following weighted Dirac  
$\delta$-measures 
\begin{equation}
F\left(  r,\theta\right)  =\frac{C_{0}}{r^{\frac{\gamma
+1}{2}+d}}\delta\left(  \theta-\theta_{0}\right)  ,\quad C_{0}>0\,, \label{C4}%
\end{equation}
where $\theta_{0}\in \Delta^{d-1}$ is fixed but arbitrary.
\end{remark}

\bigskip

\section{Main result}
\label{sec:4}

The main  result of this paper is a rigorous proof of the so-called
\textit{localization} in the above described stationary solutions. 
It turns out that  all solutions to
(\ref{S1E5}), (\ref{eq:contStat}) concentrate along a straight line as
$\left\vert x\right\vert $ or $\left\vert \alpha\right\vert $ tend to infinity
respectively. 

\begin{theorem}
\label{LocContInj}Suppose that $K$ is a continuous symmetric function that satisfies
(\ref{Ineq2}), (\ref{Ineq3}) with $\gamma+2p<1.$ 
Suppose that 
$\eta\in\mathcal{M}_{+,b}\left(  \mathbb{R}_{*}^{d}\right)  $
is supported inside the set $\defset{x\in \R^d_*}{1\le |x|\le L}$ for some $L>1$ and 
$\eta\neq 0$.
Let $f\in\mathcal{M}_{+}\left(
\mathbb{R}_{*}^{d}\right)  $ be a stationary injection solution of
(\ref{eq:contStat}) in the sense of Definition \ref{StatInjDef} and let $F$ be
defined by means of (\ref{G2}). Then, there exists $b \in (0,1)$ and a function 
$\delta
:\mathbb{R}_{*}\rightarrow\mathbb{R}_{+}$ with $\lim_{R\rightarrow\infty
}\delta\left(  R\right)  =0$ such that 
\begin{equation}
\lim_{R\rightarrow\infty}\left(  \frac{\int_{\left[  R,R/b\right]  }%
dr\int_{\Delta^{d-1}\cap\left\{  \left\vert \theta-\theta_{0}\right\vert
\leq\delta\left(  R\right)  \right\}  }d\tau\left(  \theta\right)  F\left(
r,\theta\right)  }{\int_{\left[  R,R/b\right]  } dr\int_{\Delta^{d-1}}d\tau\left(
\theta\right)  F\left(  r,\theta\right)  }\right)  =1 \label{S2E4}%
\end{equation}
where 
\begin{equation}
\theta_{0}=\frac{\int_{\mathbb{R}_{*}^{d}}x\eta\left(  x\right)  dx}%
{\int_{\mathbb{R}_{*}^{d}}\left\vert x\right\vert \eta\left(  x\right)  dx} \in \Delta^{d-1}\,.
\label{S2E5}%
\end{equation}

\end{theorem}

\bigskip

Notice that Theorem \ref{LocContInj} implies that any stationary injection
solution to (\ref{eq:contStat}) concentrates for large values of $\left\vert
x\right\vert $ along the direction $x=\left\vert x\right\vert \theta_{0}.$ A
similar result holds for the discrete problem, since the solutions to
(\ref{S1E5}) are solutions to (\ref{eq:contStat}) supported in the points with
integer coordinates $\mathbb{N}_{*}^{d}  .$
Although the following  result is a Corollary of Theorem \ref{LocContInj},  we have preferred to
formulate it separately due to the fact that the discrete coagulation
equations have an independent interest in many applications in  aerosol science. It is worth to remark that an
analogue of the solution \eqref{C4} does not exist for the discrete
Smoluchowski coagulation equation.

\begin{theorem}
\label{LocDiscInj}Suppose that $K$ is symmetric and satisfies (\ref{Ineq1}), (\ref{Ineq3}) with
$\gamma+2p<1.$ 
Let $\left\{  s_{\alpha}\right\}  _{\alpha\in\mathbb{N}_{*}^{d}}$ 
be a nonnegative sequence supported on a
finite number of values $\alpha$, and assume the sequence is not identically zero. Suppose that
$\left\{  n_{\alpha}\right\}  _{\alpha\in\mathbb{N}_{*}^{d}  }$ 
is a stationary injection solution to (\ref{S1E5}) in the sense
of Definition \ref{DefStatInjDis}. Then, there exists $b \in(0,1)$ and a function
$\delta:\mathbb{R}_{*}\rightarrow\mathbb{R}_{+}$ with $\lim_{R\rightarrow
\infty}\delta\left(  R\right)  =0$ such that 
\begin{equation}
\lim_{R\rightarrow\infty}\left(  \frac{\sum_{\left\{  R\leq\left\vert
\alpha\right\vert \leq R/b\right\}  \cap\left\{  \left\vert \frac{\alpha
}{\left\vert \alpha\right\vert }-\theta_{0}\right\vert \leq\delta\left(
R\right)  \right\}  }n_{\alpha}}{\sum_{\left\{  R\leq\left\vert \alpha
\right\vert \leq R/b\right\}  }n_{\alpha}}\right)  =1 \label{S2E6}%
\end{equation}
where 
\begin{equation}
\theta_{0}=\frac{\sum_{\alpha\in\mathbb{N}_{*}^{d}
}s_{\alpha}\alpha}{\sum_{\alpha\in\mathbb{N}_{*}^{d} }s_{\alpha}\left\vert \alpha\right\vert }\in \Delta^{d-1}\,. \label{S2E7}%
\end{equation}

\end{theorem}
\bigskip

We will prove also that all constant flux solutions to (\ref{S1E7}) are 
supported along a half-line $\left\{  x=r\theta_{0}:r>0\right\}  $ and their
analysis can be reduced to the constant flux solutions in the one-component
case $d=1.$ More precisely we have:
\begin{theorem}
\label{LocConstFlux}
Suppose that $K$ is a continuous symmetric function that satisfies
(\ref{Ineq2}), (\ref{Ineq3}) with $\gamma+2p<1$. 
Suppose that $f\in
\mathcal{M}_{+}\left(  \mathbb{R}_{*}^{d}\right)  $ is a constant flux
solution to (\ref{S1E7}) in the sense of Definition \ref{DefConstFlux} with total flux $J_0 > O$. We 
define $F$ as in (\ref{G2}) and $G$ as in (\ref{G1}).   Then, 
\[
F\left(  r,\theta\right)  =\frac{H\left(  r\right)  }{r^{d-1}}\delta\!\left(
\theta-\theta_{0}\right)
\]
where $\theta_0 = J_0/|J_0| \in \Delta^{d-1}$ and $H\in\mathcal{M}_{+}\left(\mathbb{R}_{*}\right)$ is a constant flux solution for the 
one-component
coagulation equation with the kernel $\tilde{K}\left(  r,\rho\right)  =G\left(
r,\rho;\theta_{0},\theta_{0}\right)  .$
\end{theorem}

\begin{remark}
We observe that we did not require the kernel $K$ to be homogeneous  
 for the asymptotic localization of 
stationary injection solutions (Theorems \ref{LocContInj}, \ref{LocDiscInj}) nor to prove the complete localization of constant flux solutions in Theorem \ref{LocConstFlux}. Notice that, if the kernel $K$ is homogeneous, there are constant flux solutions to \eqref{S1E7} with $H\left(  r\right)=r^{-\frac{\gamma+3}{2}}$ (cf.~\eqref{C4}). However, we cannot expect all the constant flux solutions to the one-dimensional coagulation equation to be power-law solutions,  even for homogeneous kernels.  
Indeed, there are one-dimensional coagulation kernels for which there
exist constant flux solutions which are not power laws (cf.~\cite{FLNV3}).
\end{remark}

\bigskip

The proofs (cf.\ Sections \ref{sec:localization} and \ref{sec:loca_constantFlux})  rely on  growth bounds for the stationary injection and constant flux solutions given in Proposition \ref{UppEstCont} below which has been proven in \cite{FLNV1}. The bounds  are used to   derive  estimates for an appropriate family in $R$ of probability measures in $\theta$ as defined in \eqref{S7E8} in Section \ref{sec:localization}. These measures are then shown to converge weakly to  the Dirac  $\delta(\theta-\theta_0)$ as $R\to \infty$ by using a measure concentration estimate (cf.\ Lemma \ref{LemmSplit}). 
In the case of  constant flux solutions $F(r,\theta)$ we further show that the mass is concentrated not only asymptotically but   for all $r$ by using  similar  estimates.

\bigskip


\section{Technical tools}
\label{sec:8}

\bigskip

In this Section we collect some results that will be used later in the proofs of  Theorems \ref{LocContInj}, \ref{LocDiscInj}, and \ref{LocConstFlux}. Some statements are taken from \cite{FLNV} and \cite{FLNV1}.

\subsection{Reduction of the problem to $p=0$}
\label{sec:reductionp}

The kernels $K$ satisfying (\ref{Ineq1})--(\ref{Ineq3}) can be
characterized by the two parameters $\gamma, p$.
Using a suitable change of variable, we can transform the stationary  solutions  with kernels $K$ 
into stationary solutions  with new kernels $\tilde{K}$  with parameters
$\tilde{\gamma}=\gamma+2p$ and $\tilde{p}=0$.

This follows from an idea introduced in \cite{DP} (see also \cite{BLL})\footnote{We thank P.\
Laurencot for conveying the idea and related references to us; a slightly different combination of the scaling idea with techniques of \cite{FLNV} may be found in the recent paper \cite{Laurencotpreprint20}.}. The idea is based on the observation that  $f$ is a stationary injection solution (resp. constant flux solution) associated to the kernel $K$ if and only if $h :=\left\vert x\right\vert ^{-p}f$ is a stationary injection solution (resp. constant flux solution) associated with the kernel $\tilde K$ defined by
\begin{equation}
\tilde{K}\left(  x,y\right)  =K\left(  x,y\right)  \left\vert x\right\vert
^{p}\left\vert y\right\vert ^{p}. \label{S2E1}%
\end{equation}
This new kernel $\tilde K$ satisfies (\ref{Ineq2}) after
replacing $\gamma$ by $\tilde{\gamma
}=\gamma+2p$ and $p$ by $\tilde{p}=0$, ie., $\tilde K$ satisfies 
\begin{equation}
c_{1}\left(  \left\vert x\right\vert +\left\vert y\right\vert \right)
^{\gamma+2p}\leq\tilde{K}\left(  x,y\right)  \leq c_{2}\left(  \left\vert
x\right\vert +\left\vert y\right\vert \right)  ^{\gamma+2p} \,,  \qquad x,y\in \R^d_*\,.\label{S2E2}%
\end{equation}
We then have the following Lemma (cf. \cite{FLNV1}).

\begin{lemma}\label{EquivDefin}
 Let  $\eta\in\mathcal{M}_{+,b}\left(  \mathbb{R}_{*}^{d}\right)  $ be
supported inside the set $\defset{x\in \R^d_*}{1\le |x|\le L}$ for some $L>1$. 
Suppose that $K$ is continuous and it satisfies (\ref{Ineq2}), (\ref{Ineq3}).
The Radon measure $f\in\mathcal{M}_{+}\left(  \mathbb{R}_{*}^{d}\right)  $ is
a stationary injection solution to (\ref{eq:contStat}) in the sense of Definition \ref{StatInjDef} (resp. a constant flux solution to (\ref{S1E7}) in the sense of Definition \ref{DefConstFlux}) if and only if the Radon measure $h\left(x\right)  =\frac{f\left(  x\right)  }{\left\vert x\right\vert ^{p}}$ is a
stationary injection solution to (\ref{eq:contStat}) (resp. a constant flux solution to (\ref{S1E7})) with the kernel
$\tilde{K}$ defined in (\ref{S2E1}). The kernel $\tilde{K}$ satisfies
(\ref{S2E2}).
\end{lemma}
\bigskip

\subsection{A technical Lemma}

The following  Lemma allows to transform estimates of averaged integrals into estimates in the whole line. This Lemma is a particular case of  Lemma 2.10 (items 2 and 3) in \cite{FLNV}. 
\begin{lemma}
\label{lem:bound}
Suppose $a>0, R \in [a,\infty]$ and $b,r \in (0,1)$ are such that $a \leq rR $. Define the interval $I=[a,R]$ if $R<\infty$, or $I = [a,\infty)$ if $R=\infty$. Consider some  $f \in \mathcal{M_+}(\R_*)$ and $\varphi \in C(\R_*)$, with $\varphi \geq 0$. If there is a polynomial function $g(x) = c_0x^q$ with $c_0\geq 0$ and $q \in \R$ such that $g \in L^1(I)$ and
\begin{equation*}\label{eq:avg}
\frac{1}{z}\int_{[bz,z]} \varphi(x) f(dx) \leq g(z)\,, \quad \text{for } z \in I \,,
\end{equation*}
then there is a constant $C>0$ which depends only on $r, b $ and $q$ such that
\begin{equation*}\label{eq:bound}
\int_{I} \varphi(x) f(dx) \leq Cc_0 \int_{I}x^q dx.
\end{equation*}
\end{lemma}

\subsection{Growth bounds}

We now recall some relevant growth bounds, obtained in \cite{FLNV1}, which are valid for any stationary injection solution of the continuous  (\ref{eq:contStat}) and discrete  (\ref{S1E5}) equations, as well as for any constant flux solution.
\begin{proposition}
\label{UppEstCont}
Suppose that $K$ is a continuous symmetric function satisfying
(\ref{Ineq2}), (\ref{Ineq3}) with $\gamma+2p<1.$ Suppose that 
$\eta\in\mathcal{M}_{+,b}\left(  \mathbb{R}_{*}^{d}\right)  $
is supported inside the set $\defset{x\in \R^d_*}{1\le |x|\le L}$ for some $L>1$ and let
$\vert J_{0}\vert $  denote  the total mass injection rate, where $J_{0} =\int
_{\mathbb{R}_{*}^{d}} x\eta\left(  dx\right)\in \R^d_*$. 

Consider a stationary injection solution $f\in\mathcal{M}
_{+}\left(  \mathbb{R}_{*}^{d}\right)  $ to (\ref{eq:contStat}) in the sense
of Definition \ref{StatInjDef}.
Then
there exist positive constants $C_{1},\ C_{2}$ and $b\in\left(  0,1\right)  $
depending only on $\gamma,\ p,\ d$ and the constants $c_{1},c_{2}$ in
(\ref{Ineq2}) such that with $\xi = \frac{L}{b}$ it holds %
\begin{align}
\frac{1}{z}\int_{\frac{z}{2}\leq\left\vert x\right\vert \leq z}f\left(
dx\right)   &  \leq\frac{C_{1}\sqrt{\vert J_{0}\vert }}{z^{\frac{3+\gamma}{2}}%
}\ \ \ \text{for all }z>0\,,\label{S3E4}\\
\frac{1}{z}\int_{bz\leq\left\vert x\right\vert \leq z}f\left(  dx\right)   &
\geq\frac{C_{2}\sqrt{\vert J_{0}\vert }}{z^{\frac{3+\gamma}{2}}}\ \ \ \text{for all }%
z> \xi \, .\label{S3E5}%
\end{align}

Alternatively, suppose that $f$ is a nontrivial constant 
flux solution as  in Definition \ref{DefConstFlux}
with flux $J_0$. Then
there exist positive constants $C_{1},\ C_{2}$ and $b\in\left(  0,1\right)  $
depending only on $\gamma,\ p,\ d$ and the constants $c_{1},c_{2}$ in
(\ref{Ineq2}) such that \eqref{S3E4} and \eqref{S3E5} hold with $\xi = 0$.
\end{proposition}

\begin{remark}
Notice that $\vert J_{0}\vert $ is the total injection rate, i.e., it 
includes all of the possible monomer types.
\end{remark}

\begin{remark}
We observe that for $\gamma> -1$, Proposition \ref{UppEstCont}  and Lemma \ref{lem:bound} imply that the
number of clusters associated to the stationary injection solutions
$\int_{{\mathbb{R}}_{*}} f(dx)$ is finite and the following integral estimates hold:
\[
\frac{ C_{1}\sqrt{\vert J_{0}\vert }}{z^{(\gamma+ 1)/2}} \leq\int_{\left\{\vert x\vert \geq z\right\}} f(dx) \leq
\frac{ C_{2}\sqrt{\vert J_{0}\vert }}{z^{(\gamma+ 1)/2}} \quad\text{for } z \geq1
\]
where $\vert J_{0}\vert = \int_{\R^d_{*}} \vert x \vert \eta(dx)$ and $0< C_{1} \leq C_{2} $.
\end{remark}

\begin{remark}
It is possible to obtain similar lower estimates for stationary injection
solutions $\left\{  n_{\alpha}\right\}  _{\alpha\in\mathbb{N}_{*}^{d}  }$ to the discrete problem (\ref{S1E5}). More precisely, we have (cf.~\cite{FLNV1}) 
\begin{align*}
\frac{1}{z}\sum_{\frac{z}{2}\leq\left\vert \alpha\right\vert \leq z}%
n_{\alpha}    \leq\frac{C_{1}\sqrt{\vert J_{0}\vert }}{z^{\frac{3+\gamma}{2}}%
}\ \ \ \text{for all }z\geq 1\,, \qquad 
\frac{1}{z}\sum_{bz\leq\left\vert \alpha\right\vert \leq z}n_{\alpha}  
\geq\frac{C_{2}\sqrt{\vert J_{0}\vert }}{z^{\frac{3+\gamma}{2}}}\ \ \ \text{for all }%
z\geq\frac{L_{s}}{b}\,. 
\end{align*}
where the constants $c_{1},c_{2},C_{1},\ C_{2}$ and $b\in\left(  0,1\right)$ are as in Proposition \ref{UppEstCont} and the total injection rate of monomers is given by $\vert J_{0}\vert 
=\sum_{\alpha}\left\vert \alpha\right\vert s_{\alpha}.$
\end{remark}
\bigskip

\subsection{A measure concentration Lemma}
 
The following Lemma states that given any probability measure $\lambda$ on the unit simplex
either a quadratic functional acting on the space of measures $\lambda$ is large, or there exists a set with small diameter containing most of the mass of the measure  $\lambda$. 

\begin{lemma}
\label{LemmSplit} There is a constant $C_d>0$ which depends only on the dimension $d\ge 1$ and for which the following alternative holds.

Suppose a Borel probability measure $\lambda\in
\mathcal{M}_{+,b}\left(  \Delta^{d-1}\right)  $
and parameters $\vep,\delta\in (0,1)$ are given.
Then at least one of the following alternatives is true:
\begin{itemize}
\item[(i)] There exists a measurable set $A\subset\Delta^{d-1}$ with
$\diam\left(  A\right)  \leq\vep$ such that $\int
_{A}\lambda(d\theta)> 1-\delta$.
\item[(ii)] $\int_{\Delta^{d-1}}\lambda\left(  d\theta\right)  \int_{\Delta^{d-1}%
}\lambda\left(  d\sigma\right)  \left\Vert \theta-\sigma\right\Vert ^{2}%
\geq C_d \delta \varepsilon^{d+1}$ where $\left\Vert \cdot\right\Vert $ denotes the Euclidean
distance in $\R^d$ given by $\left\Vert \theta\right\Vert ^{2}=\sum_{j=1}^{d}\left(
\theta_{j}\right)  ^{2}.$
\end{itemize}
\end{lemma}
\begin{proof}
If $d=1$, we have $\Delta^{d-1}=\{1\}$ so $\diam(\Delta^{d-1})=0$. Thus alternative (i) holds in this case always, with the choice $A=\Delta^{d-1}$.  We may set $C_1=1$, although it will never be used.

Let us then suppose $d\ge 2$ and that the case (i) is not true.  It suffices to prove that we can define $C_d>0$, depending only on $d$,
such that case (ii) holds.

Given $x\in \Delta^{d-1}$, let us consider its following $\R^d$-metric neighbourhood:
define
\[
 A(x;\vep) := \defset{\theta\in \Delta^{d-1}}{\norm{\theta-x}<\frac{1}{2}\vep}\,,
\]
and denote its $\lambda$-measure by
\[
 m(x;\vep) := \lambda(A(x;\vep)) = \int_{\Delta^{d-1}} \lambda(d\theta) \cf{\norm{\theta-x}<\frac{1}{2}\vep}\,.
\]
For clarity, we drop $\vep$ from the notation  in the following, we recall that it is fixed.
Then $\diam(A(x))\leq\vep$,  so we can conclude from the assumption that
$m(x)\le 1-\delta$.

Let us then consider the expectation in item (ii), and estimate it from below as follows, using generic characteristic functions:
\begin{align*}
 &\int_{\Delta^{d-1}}\lambda\left(  d\theta\right)  \int_{\Delta^{d-1}%
}\lambda\left(  d\sigma\right)  \left\Vert \theta-\sigma\right\Vert ^{2}
\ge \int_{\Delta^{d-1}}\lambda\left(  d\theta\right)  \int_{\Delta^{d-1}%
}\lambda\left(  d\sigma\right)  \left\Vert \theta-\sigma\right\Vert ^{2} 
\cf{\norm{\theta-\sigma}\ge \frac{\vep}{4}} 
\\ & \quad 
\ge \frac{\vep^2}{4^2}\int_{\Delta^{d-1}}\lambda\left(  d\theta\right)  \int_{\Delta^{d-1}%
}\lambda\left(  d\sigma\right)  
\cf{\norm{\theta-\sigma}\ge \frac{\vep}{4}} \,.
\end{align*}

Now for any $x\in \R^d$, by the triangle inequality,
assuming $\norm{\theta-x}\ge \frac{\vep}{2}$ and $\norm{x-\sigma}<  \frac{\vep}{4}$
implies $\norm{\theta-\sigma}\ge  \frac{\vep}{4}$.  Therefore,
\[
 \cf{\norm{\theta-\sigma}\ge \frac{\vep}{4}} \ge 
 \cf{\norm{\theta-x}\ge \frac{\vep}{2}} \cf{\norm{x-\sigma}<  \frac{\vep}{4}}\,.
\]
Thus we obtain the following lower bound, valid for any $x\in \Delta^{d-1}$,
\begin{align*}
 &\int_{\Delta^{d-1}}\lambda\left(  d\theta\right)  \int_{\Delta^{d-1}%
}\lambda\left(  d\sigma\right)  \left\Vert \theta-\sigma\right\Vert ^{2}
\ge \frac{\vep^2}{4^2}\int_{\Delta^{d-1}}\lambda\left(  d\theta\right) 
\cf{\norm{\theta-x}\ge \frac{\vep}{2}} 
\int_{\Delta^{d-1}}\lambda\left(  d\sigma\right)  \cf{\norm{x-\sigma}<  \frac{\vep}{4}} \,.
\end{align*}
Here $\cf{\norm{\theta-x}\ge \frac{\vep}{2}} =1-\cf{\norm{\theta-x}<\frac{\vep}{2}}$, and thus $\int_{\Delta^{d-1}}\lambda\left(  d\theta\right) 
\cf{\norm{\theta-x}\ge \frac{\vep}{2}}=1-m(x)\ge  \delta$.
We conclude that
\begin{align*}
 &\int_{\Delta^{d-1}}\lambda\left(  d\theta\right)  \int_{\Delta^{d-1}%
}\lambda\left(  d\sigma\right)  \left\Vert \theta-\sigma\right\Vert ^{2}
\ge \frac{\vep^2}{4^2} \delta 
\int_{\Delta^{d-1}}\lambda\left(  d\sigma\right)  \cf{\norm{x-\sigma}<  \frac{\vep}{4}} \,,
\quad x\in \Delta^{d-1}\,.
\end{align*} 

We now integrate the previous inequality over the measure $d\tau(x)$.  Denoting $c_d := \int_{\Delta^{d-1}}d\tau(x)$ and using Fubini's Theorem, we thus obtain a lower bound
\begin{align*}
 &\int_{\Delta^{d-1}}\lambda\left(  d\theta\right)  \int_{\Delta^{d-1}%
}\lambda\left(  d\sigma\right)  \left\Vert \theta-\sigma\right\Vert ^{2}
\ge \frac{\vep^2}{c_d 4^2} \delta 
\int_{\Delta^{d-1}}\lambda\left(  d\sigma\right) \int_{\Delta^{d-1}}d\tau(x) \cf{\norm{x-\sigma}<  \frac{\vep}{4}} \,.
\end{align*} 
We claim that there is a uniform lower bound $C_d>0$ which depends only on $d$ and such that, for any $\sigma\in \Delta^{d-1}$, 
\begin{align}\label{eq:lowergeomest}
 & \frac{1}{c_d 4^2} 
\int_{\Delta^{d-1}}d\tau(x) \cf{\norm{x-\sigma}<  \frac{\vep}{4}} \ge C_d \vep^{d-1} \,.
\end{align} 
Since $\lambda$ is a probability measure, we may then conclude that
\begin{align*}
 &\int_{\Delta^{d-1}}\lambda\left(  d\theta\right)  \int_{\Delta^{d-1}%
}\lambda\left(  d\sigma\right)  \left\Vert \theta-\sigma\right\Vert ^{2}
\ge C_d \delta \vep^{d+1} \,,
\end{align*} 
and thus item (ii) holds for this constant $C_d$ which satisfies the requirements of the Lemma.

It only remains to prove (\ref{eq:lowergeomest}). The proof relies on the following geometrical argument.
We consider the volume of the intersection of the simplex with $d$-balls centered at $\sigma$.  The worst case scenario occurs at the extreme corner points of the simplex.  However, since $d$ is fixed and finite, even the cones associated with the corner points have a finite volume fraction
of the ball, hence they have the same scaling as the radius goes to zero.  This implies that $\int_{\Delta^{d-1}}d\tau(x) \cf{\norm{x-\sigma}<  \frac{\vep}{4}} \ge c'_d \vep^{d-1}$
using some $c'_d>0$ and for all $\vep\in (0,1)$, and this estimate may be used to complete the proof of (\ref{eq:lowergeomest}). 

For the detailed proof, let us parametrize the simplex using the coordinate system introduced in Section~\ref{sec:change_var}: for $\theta\in \Delta^{d-1}$ denote $\hat{\theta}=\left(\theta_1,\ldots,\theta_{d-1}\right)\in \R_+^{d-1}$, and note that then $|\hat\theta|\le 1$ and 
$\theta_d=1-|\hat{\theta}|$.  Using the Schwarz inequality, we find
$\norm{\theta-\sigma}^2\le d\norm{\hat\theta-\hat\sigma}^2$, and thus
$\cf{\norm{\theta-\sigma}<  \frac{\vep}{4}} \ge \cf{\norm{\hat\theta-\hat\sigma}<  \frac{\vep}{4\sqrt{d}}}$.  Therefore, denoting $\vep':= \frac{\vep}{4\sqrt{d}}>0$,
\begin{align*}
& \int_{\Delta^{d-1}}d \tau\left(  \theta\right)\cf{\norm{\theta-\sigma}<  \frac{\vep}{4}} \ge \sqrt{d}\int_{\R_+^{d-1}} d \hat\theta \cf{|\hat\theta|\le 1}\cf{\norm{\hat\theta-\hat\sigma}<  \vep'}
\\ & \quad 
=\sqrt{d}(\vep')^{d-1}\int_{\R^{d-1}} d y \cf{\norm{y}< 1}
\prod_{j=1}^{d-1} \cf{y_j\ge -\frac{\hat\sigma_j}{\vep'}}
\cf{\sum_{j=1}^{d-1}y_j\le \frac{1-|\hat\sigma|}{\vep'}}\,,
\end{align*}
where we have
made a change of variables to 
$y=\frac{1}{\vep'}\left(\hat\theta-\hat\sigma\right)$.
Now if $|\hat\sigma|\le 1-\vep'\sqrt{d}=1-\frac{\vep}{4}$, the remaining integrand is one on the 
set $\defset{y}{\norm{y}<1\,, y_j\ge 0\text{ for all }j }$, and hence its value may be bounded from below by an only $d$-dependent strictly positive constant.
Otherwise, there has to be some $j'$ such that $\hat{\sigma}_{j'}\ge \frac{1}{2d}$ and 
then $\frac{\hat{\sigma}_{j'}}{\vep'}\ge \frac{2}{\sqrt{d}}$.
Now for $y$ such that $\norm{y}<1$,
$-\frac{2}{\sqrt{d}}\le y_{j'}\le -\frac{1}{\sqrt{d}}$ and $0\le y_j \le \frac{1}{d\sqrt{d}}$ for $j\ne j'$, all the conditions in the characteristic functions are satisfied, and thus the integrand is then one.  Again the integral over this non-empty set results in a lower bound by an only $d$-dependent strictly positive constant.  This proves (\ref{eq:lowergeomest}) and completes the proof of the Lemma.
\end{proof}

\bigskip

\section{Localization properties of the stationary injection solutions 
for the continuous and discrete coagulation equation}\label{sec:localization}
 
We prove here Theorems \ref{LocContInj} and \ref{LocDiscInj}. 
Given that the solutions to the discrete problem (\ref{S1E5}) can be thought of as solutions to
the continuous model (\ref{eq:contStat}) it will be sufficient to prove the
localization results in the case of the continuous equation (\ref{eq:contStat}).

We can argue as in Section \ref{sec:reductionp} in order to reduce the localization in stationary injection solutions to the case in which the kernels $K$ satisfy (\ref{Ineq2}), (\ref{Ineq3}) with
$p = 0, \gamma  < 1$. We formulate this result precisely.

\begin{lemma}
\label{LocEquiv}Suppose that $K$ and $\eta$ are as in the statement of Theorem
\ref{LocContInj}. Let $f\in\mathcal{M}_{+}\left(  \mathbb{R}_{*}^{d}\right)  $
be a stationary injection solution of (\ref{eq:contStat}) in the sense of
Definition \ref{StatInjDef} and let $F$ be defined by means of (\ref{G2}).
Define $h\left(  x\right)  =\frac{f\left(  x\right)  }{\left\vert x\right\vert
^{p}}.$ Then $h$ is a stationary injection solution  with the kernel $\tilde{K}$ in
(\ref{S2E1}) which satisfies (\ref{S2E2}). Let $H\left(  r,\theta\right)$ be 
obtained from $h\left(  x\right)$ as $F$ is obtained from $f$.  Then, there exists $b \in (0,1)$ and a function  $\delta:\mathbb{R}_{*}\rightarrow\mathbb{R}_{+}$ satisfying
$\lim_{R\rightarrow\infty}\delta\left(  R\right)  =0$ and such that
(\ref{S2E4}) holds for some $\theta_{0}\in\Delta^{d-1}$ if and only if%
\begin{equation}
\lim_{R\rightarrow\infty}\left(  \frac{\int_{\left[  R,R/b\right]  }%
dr\int_{\Delta^{d-1}\cap\left\{  \left\vert \theta-\theta_{0}\right\vert
\leq\delta\left(  R\right)  \right\}  }d\tau\left(  \theta\right)  H\left(
r,\theta\right)  }{\int_{\left[  R,R/b\right]  }\int_{\Delta^{d-1}}d\tau\left(
\theta\right)  H\left(  r,\theta\right)  }\right)  =1 \ .\label{S3E3}%
\end{equation}

\end{lemma}

\begin{proof}
First, note that (\ref{S2E4}) is equivalent to
\begin{equation}
\lim_{R\rightarrow\infty}\left(  \frac{\int_{\left[  R,R/b\right]  }%
dr\int_{\Delta^{d-1}\cap\left\{  \left\vert \theta-\theta_{0}\right\vert
>\delta\left(  R\right)  \right\}  }d\tau\left(  \theta\right)  F\left(
r,\theta\right)  }{\int_{\left[  R,R/b\right]  }\int_{\Delta^{d-1}}d\tau\left(
\theta\right)  F\left(  r,\theta\right)  }\right)  =0 \,. \label{S3E2}%
\end{equation}
Using that $H\left(  r,\theta\right)dr d\tau\left(  \theta\right)   =r^{-p}
F\left(  r,\theta\right)d r d\tau\left(  \theta\right)$ we obtain%
\begin{align*}
& \int_{\left[  R,R/b\right]  }dr\int_{\Delta^{d-1}\cap\left\{  \left\vert
\theta-\theta_{0}\right\vert >\delta\left(  R\right)  \right\}  }d\tau\left(
\theta\right)  H\left(  r,\theta\right)
\\ & \quad
\leq\frac{b^{\min\left\{
0,p\right\}} }{R^{p}}\int_{\left[  R,R/b\right]  }dr\int_{\Delta^{d-1}%
\cap\left\{  \left\vert \theta-\theta_{0}\right\vert >\delta\left(  R\right)
\right\}  }d\tau\left(  \theta\right)  F\left(  r,\theta\right)\,,
\end{align*}
and
\begin{align*}
& \int_{\left[  R,R/b\right]  }dr\int_{\Delta^{d-1}\cap\left\{  \left\vert
\theta-\theta_{0}\right\vert >\delta\left(  R\right)  \right\}  }d\tau\left(
\theta\right)  H\left(  r,\theta\right)
\\ & \quad
 \geq\frac{ b^{\max\left\{
0,p\right\}} }{R^{p}}
 \int_{\left[  R,R/b\right]  }dr\int_{\Delta^{d-1}%
\cap\left\{  \left\vert \theta-\theta_{0}\right\vert >\delta\left(  R\right)
\right\}  }d\tau\left(  \theta\right)  F\left(  r,\theta\right)\,,
\end{align*}

If we assume that (\ref{S3E2}) holds, and then use these two inequalities, we find that
\[
\lim_{R\rightarrow\infty}\left(  \frac{\int_{\left[  R,R/b\right]  }%
dr\int_{\Delta^{d-1}\cap\left\{  \left\vert \theta-\theta_{0}\right\vert
>\delta\left(  R\right)  \right\}  }d\tau\left(  \theta\right)  H\left(
r,\theta\right)  }{\int_{\left[  R,R/b\right]  }\int_{\Delta^{d-1}}d\tau\left(
\theta\right)  H\left(  r,\theta\right)  }\right)  =0\,.
\]
from which (\ref{S3E3}) follows.
The opposite direction is proven via similar estimates, and the 
rest of the Lemma is a matter of comparison of definitions.
\end{proof}

We now prove Theorem \ref{LocContInj}. \medskip

\begin{proofof}
[Proof of Theorem \ref{LocContInj}]Due to Lemma \ref{LocEquiv} it is
enough to prove the result for $p=0,\ \gamma<1.$ 
For convenience, we rewrite here the  weak formulation \eqref{weak_rtheta}, which is valid for any function
$\psi\in C_{c}^1\left( \R_*  \times \Delta^{d-1}\right)  $%
\begin{align}
&  \frac{1}{2d}\int_{\mathbb{R}_{*}}r^{d-1}dr\int_{\mathbb{R}_{*}}\rho
^{d-1}d\rho\int_{\Delta^{d-1}}d\tau\left(  \theta\right)  \int_{\Delta^{d-1}%
}d\tau\left(  \sigma\right)  G\left(  r,\rho;\theta,\sigma\right)  F\left(
r,\theta\right)  F\left(  \rho,\sigma\right) \nonumber\\
&  \times\left[  \psi\left(  r+\rho,\frac{r}{r+\rho}\theta+\frac{\rho}{r+\rho
}\sigma\right)  -\psi\left(  r,\theta\right)  -\psi\left(  \rho,\sigma\right)
\right]  +\int_{{\mathbb{R}_{*}^{d}}}\varphi\left(  x\right)  \eta\left(
dx\right)  =0\,\label{S7E2}%
\end{align}
 where $\varphi(x) = \psi(\theta,r)$ and the source  $\eta$ is supported in the set $\defset{x\in \R^d_*}{1\le |x|\le L}$ for some $L>1$.
In order to obtain estimates for the measure $F$ we consider in (\ref{S7E2}) 
test functions $\psi\left(  r,\theta;R\right)  =r\phi_{R}\left(  r\right)
\left\Vert \theta\right\Vert ^{2}$ with $\left\Vert \theta\right\Vert
^{2}=\sum_{j=1}^{d}\left(  \theta_{j}\right)  ^{2}$ and $R\ge L$.    Here, $\phi_R\in C^\infty_c(\R_*)$ is chosen as a bump function:  it is monotone increasing on $(0,1]$ and monotone decreasing on $[1,\infty)$. We also assume that  
$\phi_{R}\left( r\right)  =1$ for $1\leq r\leq R,$ and $\phi_{R}\left(  r\right)  =0$ for
$r\geq 2R$.  

Then,
\begin{align*}
&  \psi\left(  r+\rho,\frac{r}{r+\rho}\theta+\frac{\rho}{r+\rho}\sigma\right)
-\psi\left(  r,\theta\right)  -\psi\left(  \rho,\sigma\right)  \\
&  =\left(  r+\rho\right)  \phi_{R}\left(  r+\rho\right)  \left\Vert \frac
{r}{r+\rho}\theta+\frac{\rho}{r+\rho}\sigma\right\Vert ^{2}-r\phi_{R}\left(
r\right)  \left\Vert \theta\right\Vert ^{2}-\rho\phi_{R}\left(  \rho\right)
\left\Vert \theta\right\Vert ^{2}\\
&  =\phi_{R}\left(  r+\rho\right)  \left(  \frac{r^{2}}{\left(  r+\rho\right)
}\left\Vert \theta\right\Vert ^{2}+\frac{2r\rho}{\left(  r+\rho\right)
}\left(  \theta\cdot\sigma\right)  +\frac{\rho^{2}}{\left(  r+\rho\right)
}\left\Vert \sigma\right\Vert ^{2}\right)  -r\phi_{R}\left(  r\right)
\left\Vert \theta\right\Vert ^{2}-\rho\phi_{R}\left(  \rho\right)  \left\Vert
\sigma\right\Vert ^{2}\\
&  =\frac{1}{\left(  r+\rho\right)  }\left[  \phi_{R}\left(  r+\rho\right)
r^{2}-\phi_{R}\left(  r\right)  r\left(  r+\rho\right)  \right]  \left\Vert
\theta\right\Vert ^{2}
+\frac{2r\rho}{\left(  r+\rho\right)  }\phi_{R}\left(  r+\rho\right)
\left(  \theta\cdot\sigma\right)\\ &  \quad 
+\frac{1}{\left(  r+\rho\right)  }\left[  \phi
_{R}\left(  r+\rho\right)  \rho^{2}-\phi_{R}\left(  \rho\right)  \rho\left(
r+\rho\right)  \right]  \left\Vert \sigma\right\Vert ^{2}
\,.
\end{align*}

We next assume $r,\rho\ge 1$, and note that then $\phi_{R}\left(
r\right)  \geq\phi_{R}\left(  r+\rho\right)  $ and $\phi_{R}\left(
\rho\right)  \geq\phi_{R}\left(  r+\rho\right)$.  This yields an upper bound
\begin{align*}
&  \psi\left(  r+\rho,\frac{r}{r+\rho}\theta+\frac{\rho}{r+\rho}\theta\right)
-\psi\left(  r,\theta\right)  -\psi\left(  \rho,\sigma\right)  \\
&  \leq\frac{1}{\left(  r+\rho\right)  }\left[  \phi_{R}\left(  r+\rho\right)
r^{2}-\phi_{R}\left(  r+\rho\right)  r\left(  r+\rho\right)  \right]
\left\Vert \theta\right\Vert ^{2}\\
&  \quad +\frac{1}{\left(  r+\rho\right)  }\left[
\phi_{R}\left(  r+\rho\right)  \rho^{2}-\phi_{R}\left(  r+\rho\right)
\rho\left(  r+\rho\right)  \right]  \left\Vert \sigma\right\Vert ^{2}+\frac{2r\rho}{\left(  r+\rho\right)  }\phi_{R}\left(  r+\rho\right)
\left(  \theta\cdot\sigma\right)  \\
&  =-\frac{\phi_{R}\left(  r+\rho\right)  r\rho}{\left(  r+\rho\right)
}\left[  \left\Vert \theta\right\Vert ^{2}+\left\Vert \sigma\right\Vert
^{2}-2\left(  \theta\cdot\sigma\right)  \right]  =-\frac{\phi_{R}\left(
r+\rho\right)  r\rho}{\left(  r+\rho\right)  }\left\Vert \theta-\sigma
\right\Vert ^{2}.
\end{align*}
Then, for $R\geq L$, we obtain from (\ref{S7E2})
\begin{align}\label{S7E3}
&  \int_{\mathbb{R}_{*}}r^{d-1}dr\int_{\mathbb{R}_{*}}\rho^{d-1}d\rho
\int_{\Delta^{d-1}}d\tau\left(  \theta\right)  \int_{\Delta^{d-1}}d\tau\left(
\sigma\right)  G\left(  r,\rho;\theta,\sigma\right)  F\left(  r,\theta\right)
F\left(  \rho,\sigma\right)  
\nonumber \\ &  \qquad\times
\frac{\phi_{R}\left(  r+\rho\right)  r\rho
}{2\left(  r+\rho\right)  }\left\Vert \theta-\sigma\right\Vert ^{2}%
\nonumber \\ &  \quad
\leq d\int_{{\mathbb{R}_{*}^{d}}}\left\vert x\right\vert \phi_{R}\left(
\left\vert x\right\vert \right)  \eta\left(  dx\right)=  d|J_0|\,,
\end{align}
where $J_0:= \int_{{\mathbb{R}_{*}^{d}}} x \eta\left(  dx\right)\in \R^d_*$ is the  mass injection rate vector and the last equality is obtained by taking into account that the injection solutions do not
have any support on values with $|x|<1$.

Since $\cf{1\le r+\rho\le R}\le \phi_{R}\left(  r+\rho\right)$, the bound in (\ref{S7E3}) remains valid after we replace $\phi_{R}$ with this characteristic function.
The new function, however, is pointwise monotone increasing in $R$, approaching $1$ pointwise as $R\to\infty$, so we may apply monotone convergence Theorem to take the limit $R\to\infty$ inside the integral. This yields
\[
\int_{\mathbb{R}_{*}}r^{d-1}dr\int_{\mathbb{R}_{*}}\rho^{d-1}d\rho\int
_{\Delta^{d-1}}d\tau\left(  \theta\right)  \int_{\Delta^{d-1}}d\tau\left(
\sigma\right)  G\left(  r,\rho;\theta,\sigma\right)  F\left(  r,\theta\right)
F\left(  \rho,\sigma\right)  \frac{r\rho\left\Vert \theta-\sigma\right\Vert
^{2}}{\left(  r+\rho\right)  }\leq K_{0}%
\]
where $K_0:=2 d |J_0|>0$.
Using now the lower estimate in (\ref{Ineq2}), (\ref{Ineq3}) with $p=0$, as well as (\ref{G1}), we obtain%
\begin{equation}
\int_{\mathbb{R}_{*}}r^{d-1}dr\int_{\mathbb{R}_{*}}\rho^{d-1}d\rho\int
_{\Delta^{d-1}}d\tau\left(  \theta\right)  \int_{\Delta^{d-1}}d\tau\left(
\sigma\right)  \left(  r+\rho\right)  ^{\gamma-1}F\left(  r,\theta\right)
F\left(  \rho,\sigma\right)  r\rho\left\Vert \theta-\sigma\right\Vert ^{2}%
\leq\frac{K_{0}}{c_{1}}.
\label{S7E4}%
\end{equation}

For the rest of the proof, let us recall Proposition \ref{UppEstCont}, and let $b\in (0,1)$ and $C_1,C_2>0$ 
be some choice of constants for which the Proposition holds. 
We record the following bounds implied by Proposition \ref{UppEstCont} and Lemma \ref{lem:bound}: 
If $q<\frac{\gamma+1}{2}$ and $R\ge L$, we have
\[
  \frac{c_2 \sqrt{|J_0|}}{R^{\frac{\gamma+1}{2}-q}}\le \int_{\left[  R,\infty\right)\times
\Delta^{d-1}}F\left(  r,\theta\right)  r^{q+d-1}drd\tau\left(  \theta\right)\le 
 \frac{c_1 \sqrt{|J_0|}}{R^{\frac{\gamma+1}{2}-q}}\,,
\]
where the lower bound is obtained by restricting the integral into $r\in [R,b^{-1} R]$.
The constants depend on $q$ and $b$, as well as on $\gamma,\ p$ and $ d$.
In particular, for any fixed $R\ge L$ 
the value of the integral belongs to $(0,\infty)$, and we may define the Radon
probability measures
$\lambda(\theta;R)d\tau\left(  \theta\right) \in\mathcal{M}_{+,b}\left(  \Delta^{d-1}\right)  $ via the formula
\begin{equation}
\lambda\left(  \theta;R\right)  =\frac{\int_{\left[  R,\infty\right)
}F\left(  r,\theta\right)  r^{\gamma+d-1}dr}{\int_{\left[  R,\infty\right)  \times
\Delta^{d-1}}F\left(  r,\sigma\right)  r^{\gamma+d-1}drd\tau\left(  \sigma\right)
}\label{S7E8}%
\end{equation}
Let $Z(R)$ denote the value of the integral in the denominator, for which we have 
bounds
\begin{align}\label{eq:ldenombounds}
  \frac{C'_2 \sqrt{|J_0|}}{R^{\frac{1-\gamma}{2}}}\le Z(R)\le 
 \frac{C'_1 \sqrt{|J_0|}}{R^{\frac{1-\gamma}{2}}}\, .
\end{align}

To apply the measure concentration estimate, let us then consider the function
\begin{align}\label{eq:defvarla}
V(R) := \int_{\Delta^{d-1}} d\tau\left(  \theta\right)  \lambda\left(  \theta;R\right)
\int_{\Delta^{d-1}}
d\tau\left(  \sigma\right)  \lambda\left(  \sigma;R\right) \left\Vert \theta-\sigma\right\Vert ^{2}\,, \qquad R\ge L\,. 
\end{align}
Using Fubini's Theorem, we can rewrite the definition as
\begin{align*}
&  V(R) = \frac{1}{Z(R)^{2} }
 \int_{\left[  R,\infty\right)}r^{\gamma+d-1}dr\int_{\left[  R,\infty\right)}\rho^{\gamma+d-1}d\rho
\\ & \qquad \times
 \int_{\Delta^{d-1}}d\tau\left(  \theta\right)  \int_{\Delta^{d-1}}d\tau\left(
\sigma\right)  F\left(  r,\theta\right)
F\left(  \rho,\sigma\right)  \left\Vert \theta-\sigma\right\Vert ^{2}\,. 
\end{align*}
Here, $r^\gamma \rho^\gamma = 
r \rho (r+\rho)^{\gamma-1} (r^{-1}+ \rho^{-1})^{1-\gamma}\le r \rho (r+\rho)^{\gamma-1} (2/R)^{1-\gamma}$.
Therefore, by (\ref{eq:ldenombounds}), 
\begin{align*}
& V(R)\le \frac{2^{1-\gamma}}{|J_0|(C'_2)^2} \int_{\left[  R,\infty\right)}r^{d-1}dr\int_{\left[  R,\infty\right)}\rho^{d-1}d\rho\, r\rho (r+\rho)^{\gamma-1}
\\ &\qquad\times
\int_{\Delta^{d-1}}d\tau\left(  \theta\right)  \int_{\Delta^{d-1}}d\tau\left(
\sigma\right)  F\left(  r,\theta\right)
F\left(  \rho,\sigma\right)  \left\Vert \theta-\sigma\right\Vert ^{2}\,.
\end{align*}
Hence, there is a constant $C'>0$ such that 
\begin{align*}
& 0\le V(R) \le \frac{C'}{K_0} \int_{\R_ *}r^{d-1}dr\int_{\R_*}\rho^{d-1}d\rho\,  \cf{r,\rho \ge R}
\\ &\qquad\times
\int_{\Delta^{d-1}}d\tau\left(  \theta\right)  \int_{\Delta^{d-1}}d\tau\left(
\sigma\right) (r+\rho)^{\gamma-1} F\left(  r,\theta\right)
F\left(  \rho,\sigma\right)  r\rho \left\Vert \theta-\sigma\right\Vert ^{2}\,.
\end{align*}
The finite bound in (\ref{S7E4}) now allows to use dominated convergence Theorem to conclude that
\begin{align}\label{eq:VRtozero}
 \lim_{R\to \infty} V(R)=0\,.
\end{align}

We can then apply Lemma \ref{LemmSplit} as follows, with $C_d$ denoting the constant in the Lemma.  
We first choose $R_0\ge L,\frac{2}{C_d}$ so that $V(R)< \frac{C_d}{2}$ for all $R\ge R_0$, which is possible due to (\ref{eq:VRtozero}).
Given $R\ge R_0$,
we choose $\delta(R)$ and $\vep(R)$ so that the alternative (ii) fails.
To this end, let us pick an arbitrary function $\vep_0:\R_*\to \R_*$ which satisfies $\vep_0(R)\le \frac{1}{R}$ for all $R>0$, and then 
define for $R\ge R_0$
\begin{align}\label{eq:defndelvep}
 \delta(R) := \vep(R)\,,  \qquad
 \vep(R) := \left(\frac{V(R)+\vep_0(R)}{C_d}\right)^{\frac{1}{d+2}}\in (0,1)\,,
\end{align}
so that always $C_d\delta(R)\vep(R)^{d+1}=V(R)+\vep_0(R)>V(R)$ and $\delta(R)=\vep(R)\in (0,1)$, with $\vep(R)\to 0$ as $R\to\infty$. 
Then Lemma \ref{LemmSplit} implies that for each $R\ge R_0$ there 
exists $\sigma_{0}\left(  R\right)  \in\Delta^{d-1}$ for which
\[
\int_{\bar B_{\varepsilon(R)}\left(  \sigma_{0}\left(  R\right)  \right)  }\lambda
\left(  \theta;R\right)  d\tau\left(  \theta\right)  > 1-\varepsilon(R)\,.
\]

Then, the compactness of $\Delta^{d-1}$ implies that for every sequence
$\left\{  R_{n}\right\}  _{n\in\mathbb{N}}$ with $\lim_{n\rightarrow\infty
}R_{n}=\infty,$ there exists a subsequence $\left\{  R_{n_{k}}\right\}
_{k\in\mathbb{N}}$ and $\theta_{0}\in\Delta^{d-1}$ such that
$\bar\sigma_k:=\sigma_0(R_{n_k})\to \theta_0$ in $\R^d$ and, hence,
\begin{equation}
\lambda_k\left(  \theta\right):= \lambda\left(  \theta;R_{n_{k}}\right)  \rightharpoonup\delta\left(
\theta-\theta_{0}\right)  \ \ \text{as\ \ }k\rightarrow\infty\label{S7E9}%
\end{equation}
where the convergence takes place in the weak topology of measures.
The Dirac $\delta$-measure above is defined with respect to $d\tau\left(  \theta\right) $, for instance, $\int_{\Delta^{d-1}}d\tau\left(  \theta\right)\delta\left(
\theta-\theta_{0}\right) = 1$  even if $\theta_0$ is on the boundary of $\Delta^{d-1}$.

In principle, $\theta_0$ depends on the chosen subsequence $\{R_{n_{k}}\}$. We prove now that this is not the case and characterize $\theta_{0}.$ We plug into the
weak formulation (\ref{S7E2}) the test functions $\psi_j\left(  r,\theta\right)
=r\phi_{R}\left(  r\right)  \theta_j$, $j=1,2,\ldots,d$, where $\phi_R$ satisfies all earlier conditions.
In addition, we now require that there is $c>0$ such that 
$0\ge \phi'_R(r)\ge -\frac{c}{R}$ for $r\ge 1$,
so that $0\le \phi_R(r)-\phi_R(r+\rho)\le \frac{c}{R} \rho$ for all $r,\rho\ge 1$.  A function satisfying all the requirements can be constructed for instance by taking a smooth non-increasing function $g_1$ for which $g_1(r)=1$ for $r\le 1$, $g_1(r)=0$, $r\ge 2$, and then defining $\phi_R(r)=(1-g_1(2 r)) g_1(r/R)$.

Then, using a vector notation, 
\begin{align*}
&  \psi\left(  r+\rho,\frac{r}{r+\rho}\theta+\frac{\rho}{r+\rho}\sigma\right)
-\psi\left(  r,\theta\right)  -\psi\left(  \rho,\sigma\right)  \\
&  =\phi\left(  r+\rho\right)  \left(  r\theta+\rho\sigma\right)
-r\phi\left(  r\right)  \theta-\rho\phi\left(  \rho\right)  \sigma\\
&  =r\left[  \phi\left(  r+\rho\right)  -\phi\left(  r\right)  \right]
\theta+\rho\left[  \phi\left(  r+\rho\right)  -\phi\left(  \rho\right)
\right]  \sigma.
\end{align*}
Plugging this identity into (\ref{S7E2}) and using a symmetrization argument
to exchange the variables $\left(  r,\theta\right)  \longleftrightarrow\left(
\rho,\sigma\right)  $ we obtain%
\begin{align}
&\int_{\mathbb{R}_{*}}r^{d-1}dr\int_{\mathbb{R}_{*}}\rho^{d-1}d\rho\int
_{\Delta^{d-1}}d\tau\left(  \theta\right)  \int_{\Delta^{d-1}}d\tau\left(
\sigma\right)  G\left(  r,\rho;\theta,\sigma\right)  F\left(  r,\theta\right)
F\left(  \rho,\sigma\right)\times \nonumber\\&  \quad \times r\left[  \phi_{R}\left(  r\right)  -\phi_{R}\left(  r+\rho\right)  \right]  \theta= J_0d\label{S8E1}%
\end{align}
where we assume that $R>L$ and $J_0$ is the injection rate vector, as defined above.

We then expand using $\theta=\theta_0 + \theta - \theta_0$, and denote
\begin{align}\label{eq:defomegar}
&\omega(R) :=\int_{\mathbb{R}_{*}}r^{d-1}dr\int_{\mathbb{R}_{*}}\rho^{d-1}d\rho\int
_{\Delta^{d-1}}d\tau\left(  \theta\right)  \int_{\Delta^{d-1}}d\tau\left(
\sigma\right)  G\left(  r,\rho;\theta,\sigma\right)  F\left(  r,\theta\right)
F\left(  \rho,\sigma\right)\times \nonumber\\&  \quad \times r\left[  \phi_{R}\left(  r\right)  -\phi_{R}\left(  r+\rho\right)  \right]  (\theta-\theta_0)\, .%
\end{align}
We also note that
\begin{align}
&\int_{\mathbb{R}_{*}}r^{d-1}dr\int_{\mathbb{R}_{*}}\rho^{d-1}d\rho\int
_{\Delta^{d-1}}d\tau\left(  \theta\right)  \int_{\Delta^{d-1}}d\tau\left(
\sigma\right)  G\left(  r,\rho;\theta,\sigma\right)  F\left(  r,\theta\right)
F\left(  \rho,\sigma\right)\times \nonumber\\&  \quad \times r\left[  \phi_{R}\left(  r\right)  -\phi_{R}\left(  r+\rho\right)  \right]  = |J_0|d\,,
\end{align}
as seen by summing over the components in (\ref{S8E1}), since $\theta\in \Delta^{d-1}$.  Therefore, for all $R>L$,
\[
 d |J_0| \theta_0 + \omega(R) = d J_0\,.
\]
We claim that there is a sequence $R'_k> L$ such that 
$\omega(R'_k)\to 0$ in $\R^d$ as $k\to \infty$.  This implies
\[
 \theta_0 = \frac{1}{|J_0|}J_0\,,
\]
as claimed in the Theorem.  In particular, the value does not then depend on the choice of the subsequence $(R_{n_k})$ above.
Therefore, $\sigma(R)\to\theta_0$ as $R\to \infty$, and 
we can then replace (\ref{S7E9}) by $\lambda\left(  \theta;R\right)
\rightharpoonup\delta\left(  \theta-\theta_{0}\right)  $ as $R\rightarrow
\infty$.

Thus it only remains to construct the sequence $R'_k>L$, $k\in \N_+$, such that 
$\omega(R'_k)\to 0$ as $k\to \infty$.  We do this by showing that
$\lim_{\delta\to 0}\limsup_{k\to \infty}\norm{\omega(R_{n_k} \delta^{-2})}=0$
where $\delta\in \left(0,\frac{1}{4}\right]$ since then a diagonal construction allows finding suitable $R'_k$.

We start from (\ref{eq:defomegar}) after choosing arbitrary $k$ and
$\delta\in \left(0,\frac{1}{4}\right]$ and setting $R=R_{n_k} \delta^{-2}>L$.
For $r,\rho\ge 1$, we
have $0\le \phi_R(r)-\phi_{R}(r+\rho)\le \min\left(1,\frac{c}{R} \rho\right)$, by assumption.  In addition, the difference is zero, if $r\ge 2 R$ or $r+\rho\le R$.
Therefore, the integration region may be restricted to 
\[
 \Omega := \defset{(r,\rho)\in \R_*^2}{1\le r\le 2R\,,\ \rho\ge 1\,,\ r+\rho>R}\,.
\]
We now split the region into three parts $\Omega_i$, $i=1,2,3$, analogously to what was used in the proof of Proposition $6.3$ in \cite{FLNV}:
We set 
\begin{align}\label{defomisplit}
 & \Omega_1 := \defset{(r,\rho)\in \Omega}{\rho>r/\delta}\,, \qquad
 \Omega_2 := \defset{(r,\rho)\in \Omega}{\delta R\le \rho\le r/\delta}\,,\nonumber\\
 & 
 \Omega_3 := \defset{(r,\rho)\in \Omega}{\rho<\delta R}\,.
\end{align}
and then define
\begin{align} \label{I1}
& I_i =\int_{\Omega_i} dr\,d\rho \ r^{d} \rho^{d-1}
\int_{\Delta^{d-1}}d\tau\left(  \theta\right)  \int_{\Delta^{d-1}}d\tau\left(
\sigma\right)  G\left(  r,\rho;\theta,\sigma\right)  F\left(  r,\theta\right)
F\left(  \rho,\sigma\right)\times \nonumber\\&  \quad \times \left[  \phi_{R}\left(  r\right)  -\phi_{R}\left(  r+\rho\right)  \right]  (\theta-\theta_0)\, .%
\end{align}
so that $\omega(R) = \sum_{i=1}^3 I_i$.

In the region $\Omega_1$  we use the trivial bound $\left[  \phi_{R}\left(  r\right)  -\phi_{R}\left(  r+\rho\right)  \right]  \norm{\theta-\theta_0}\le 2$ 
in the integrand, as well as the estimate $G\left(  r,\rho;\theta,\sigma\right) \leq c_2 (r+\rho)^\gamma$.
By adapting the  proof of the corresponding estimate in Lemma 6.1 in \cite{FLNV} applied to the measure 
$ \int_{\Delta^{d-1}}d\tau\left(  \theta\right)   r^{d-1} F\left(  r,\theta\right)$  with the kernel  $ (r+\rho)^\gamma $, we may conclude that $\sup_{R} \|I_1(\delta,R) \| \to 0$ as $\delta \to 0$.

In the region $\Omega_3$ we use the upper bound $\left[  \phi_{R}\left(  r\right)  -\phi_{R}\left(  r+\rho\right)  \right]  \norm{\theta-\theta_0}\le 2 c \rho/R$.
Moreover, using that $\Omega_3 \subset [ 1,\delta R] \times [ R/2 , 2R]$ as well as the upper bound of the kernel in (\ref{Ineq2}), (\ref{Ineq3}) with $p=0$, and recalling  (\ref{G1}), we obtain the bound
\begin{align*}
& \norm{I_3}\le C R^{\gamma-1}\int_{\{1\le \rho \le \delta R\,,\, \frac{1}{2} R\le r\le 2R\}} dr\,d\rho \ r^{d} \rho^{d}
\int_{\Delta^{d-1}}d\tau\left(  \theta\right)  \int_{\Delta^{d-1}}d\tau\left(
\sigma\right)   F\left(  r,\theta\right)
F\left(  \rho,\sigma\right)
\\ & \quad 
\le C |J_0| \delta^{\frac{1-\gamma}{2}}
\,,
\end{align*}
where in the second inequality we have applied Lemma \ref{lem:bound} with $q=1-\frac{\gamma+3}{2}$. Also this bound is independent of $R$ and goes to zero as $\delta\to 0$.

It remains to estimate $I_2$.
 Now on $\Omega_2$ we have $(r+\rho)^\gamma\le 3^{|\gamma|}\delta^{-|\gamma|}\rho^\gamma$, and also
$\Omega_2 \subset [\delta^2 R,2 R]\times [\delta R,\infty)$.  Therefore, we obtain
an estimate
\begin{align*}
&\norm{I_2}\le C \delta^{-|\gamma|} R \int_{[\delta^2 R,2R]} dr\, r^{d-1}
\int_{[\delta R,\infty)}  d\rho \, \rho^{\gamma+d-1}
\\ & \qquad \times
\int_{\Delta^{d-1}}d\tau\left(  \theta\right)  \int_{\Delta^{d-1}}d\tau\left(
\sigma\right)   F\left(  r,\theta\right)
F\left(  \rho,\sigma\right)  \norm{\theta-\theta_0}\,.
\\ & \quad 
\le
 C' \delta^{-3 |\gamma|} R^{1-\gamma} Z(\delta^2 R) Z(\delta R) 
\int_{\Delta^{d-1}} d\theta \lambda(\theta;\delta^2 R) \norm{\theta-\theta_0}
\\ & \quad 
\le
 C'|J_0| \delta^{-3 |\gamma|+ 3/2(\gamma-1)} 
\int_{\Delta^{d-1}} d\theta \lambda(\theta;\delta^2 R) \norm{\theta-\theta_0}\,.
\end{align*}
By our choice of $R, \delta$, we have here $\delta^2 R=R_{n_k}$.
Therefore,
\[
 \int_{\Delta^{d-1}} d\theta \lambda(\theta;\delta^2 R) \norm{\theta-\theta_0}
 \le \norm{\theta_0-\bar\sigma_k}
 + \int_{\Delta^{d-1}} d\theta \lambda(\theta;R_{n_k}) \norm{\theta-\sigma_0(R_{n_k})}\,,
\]
which goes to zero as $k\to \infty$, by construction.  

Collecting the estimates of $I_i$ together, we find
$\lim_{\delta\to 0}\limsup_{k\to \infty}\norm{\omega(R_{n_k} \delta^{-2})}=0$
which completes the proof of $\lim_{R\to \infty}\sigma_0(R)= \theta_0$.
Hence, also 
\[
 v(R):=\left(\int_{\Delta^{d-1}} d\theta \lambda(\theta;R) \norm{\theta-\theta_0}\right)^{\frac{1}{2}}
 \to 0\,,
\]
and by Chebyshev-type estimate then 
\[
 \int_{\Delta^{d-1}} d\theta \lambda(\theta;R) \cf{\norm{\theta-\theta_0}> v(R)}
 \le v(R)\to 0\,.
\]

To conclude the proof of the Theorem, let us show that this result implies the claim in the Theorem, if we choose there $\delta(R)=\sqrt{d}\,v(R)$.  We have
\begin{align*}
& \frac{\int_{\left[  R,R/b\right]  }%
dr\int_{\Delta^{d-1}\cap\left\{  \left\vert \theta-\theta_{0}\right\vert
>\delta\left(  R\right)  \right\}  }d\tau\left(  \theta\right)  F\left(
r,\theta\right)  }{\int_{\left[  R,R/b\right]  }dr\int_{\Delta^{d-1}}d\tau\left(
\theta\right)  F\left(  r,\theta\right) }
\\ & \quad
\le \frac{C Z(R)}{ R^\gamma\int_{\left[  R,R/b\right]  }dr \, r^{d-1}\int_{\Delta^{d-1}}d\tau\left(
\theta\right)  F\left(  r,\theta\right) }\int_{\Delta^{d-1}} d\theta \lambda(\theta;R) \cf{\norm{\theta-\theta_0}> v(R)}
\end{align*}
where the constant $C$ depends only on $b$, $d$ and $\gamma$.  Employing 
Proposition \ref{UppEstCont} and (\ref{eq:ldenombounds}),
we find that the first factor 
on the right hand side is uniformly bounded in $R$, 
and thus the right hand side goes to zero as $R\to \infty$.
Thus  (\ref{S2E4}), (\ref{S2E5}) hold.  This concludes the proof of Theorem \ref{LocContInj}.
\end{proofof}

\begin{remark}
We notice that, in principle, the value of $\theta_0$ can be at the
boundary of $\Delta^{d-1}$. 
\end{remark}

\bigskip

We now prove Theorem \ref{LocDiscInj} for the discrete coagulation equation. 

\bigskip

\begin{proofof}
[Proof of Theorem \ref{LocDiscInj}]It is an easy consequence of Theorem
\ref{LocContInj} using the fact that $f\left(  \cdot\right)  =\sum_{\alpha
}n_{\alpha}\delta\left(  \cdot-\alpha\right)  $ and $\eta\left(  \cdot\right)
=\sum_{\alpha}s_{\alpha}\delta\left(  \cdot-\alpha\right)  $ satisfy all the
assumptions in Theorem \ref{LocContInj}. Then (\ref{S2E6}) is just a
consequence of (\ref{S2E4}).
\end{proofof}

\bigskip

\section{Localization properties of the constant flux solutions}\label{sec:loca_constantFlux}

We now prove Theorem \ref{LocConstFlux}. The assumptions and some of the details of the proof are very similar to the earlier cases, and then they will not be repeated here.

\bigskip

\begin{proofof}
[Proof of Theorem \ref{LocConstFlux}]Due to Lemma \ref{EquivDefin}  it
is enough to prove the result for $\gamma<1,\ p=0.$ We consider the weak
formulation (\ref{S7E2}) (with $\eta=0$), namely 
\begin{align}
&  \frac{1}{2}\int_{\mathbb{R}_{*}}r^{d-1}dr\int_{\mathbb{R}_{*}}\rho
^{d-1}d\rho\int_{\Delta^{d-1}}d\tau\left(  \theta\right)  \int_{\Delta^{d-1}%
}d\tau\left(  \sigma\right)  G\left(  r,\rho;\theta,\sigma\right)  F\left(
r,\theta\right)  F\left(  \rho,\sigma\right)  \times\nonumber\\
&  \times\left[  \psi\left(  r+\rho,\frac{r}{r+\rho}\theta+\frac{\rho}{r+\rho
}\sigma\right)  -\psi\left(  r,\theta\right)  -\psi\left(  \rho,\sigma\right)
\right]   =0\label{weaknoeta}. 
\end{align}
We now choose continuous test
functions $\psi\left(  r,\theta\right)  =r\phi_{R_{1},R_{2}}\left(  r\right)
\left\Vert \theta\right\Vert ^{2}$ where we require that
$0<R_{1}<R_{2}$ and we construct $\phi_{R_{1},R_{2}%
}\left(  r\right)  $ as%
\begin{equation}\label{eq:ffiR12}
\phi_{R_{1},R_{2}}=\phi_{R_{2}}-\phi_{\frac{R_{1}}{2}}%
\end{equation}
where the functions $\phi_{R}$ are defined by 
\begin{equation*}
\phi_{R}(r)=\left\{
 \begin{array}
 [c]{cc}%
 \vspace{2mm} 1 & \text{if\ \ }
 r\le  R\,,
 \\
 1-\frac{r-R}{R} & \qquad\text{if\ \ } R\leq r\leq 2R\, , \\
 0 & \text{if\ \ }
 r\geq 2 R\,. 
 \end{array}
\right.
\end{equation*}
Then each $\phi_{R}$ is Lipschitz continuous and decreasing, and $\phi_{R_{1},R_{2}}$ are compactly supported. 
Although $\phi_{R}$ is not continuously differentiable, and hence strictly speaking not an allowed test function, a standard approximation argument shows that, nevertheless, equation (\ref{weaknoeta}) holds also for this choice.
Alternatively, one may use below the smooth bump function  $\phi_{R}$
constructed after (\ref{S7E9}), although some of the constants in the upper bounds would need to be increased then.

For this choice of test functions we obtain 
\begin{align*}
&  \psi\left(  r+\rho,\frac{r}{r+\rho}\theta+\frac{\rho}{r+\rho}\sigma\right)
-\psi\left(  r,\theta\right)  -\psi\left(  \rho,\sigma\right)  \\
&  =\left(  r+\rho\right)  \phi_{R_{1},R_{2}}\left(  r+\rho\right)  \left\Vert
\frac{r}{r+\rho}\theta+\frac{\rho}{r+\rho}\sigma\right\Vert ^{2}-r\phi
_{R_{1},R_{2}}\left(  r\right)  \left\Vert \theta\right\Vert ^{2}-\rho
\phi_{R_{1},R_{2}}\left(  \rho\right)  \left\Vert \sigma\right\Vert ^{2}\\
&  =\frac{1}{\left(  r+\rho\right)  }\left[  \phi_{R_{1},R_{2}}\left(r+\rho\right)  r^{2}-\phi_{R_{1},R_{2}}\left(  r\right)  r\left(
r+\rho\right)  \right]  \left\Vert \theta\right\Vert ^{2} \\
& \quad +\frac{1}{\left(
r+\rho\right)  }\left[  \phi_{R_{1},R_{2}}\left(  r+\rho\right)  \rho^{2}%
-\phi_{R_{1},R_{2}}\left(  \rho\right)  \rho\left(  r+\rho\right)  \right]
\left\Vert \sigma\right\Vert ^{2}
+\frac{2r\rho}{\left(  r+\rho\right)  }\phi_{R_{1},R_{2}}\left(
r+\rho\right)  \left(  \theta\cdot\sigma\right)  \\
&  =-\frac{\phi_{R_{1},R_{2}}\left(  r+\rho\right)  r\rho}{\left(
r+\rho\right)  }\left\Vert \theta-\sigma\right\Vert ^{2}+\left(  \phi
_{R_{1},R_{2}}\left(  r+\rho\right)  -\phi_{R_{1},R_{2}}\left(  r\right)
\right)  r\left\Vert \theta\right\Vert ^{2}\\&\quad +\left(  \phi_{R_{1},R_{2}}\left(
r+\rho\right)  -\phi_{R_{1},R_{2}}\left(  \rho\right)  \right)  \rho\left\Vert
\sigma\right\Vert ^{2}.
\end{align*}
Plugging this into \eqref{weaknoeta} and using also the
symmetrization $\left(  r,\theta\right)  \longleftrightarrow\left(
\rho,\sigma\right)  $ we obtain
\begin{align}\label{eq:U12}
&\frac 1 2 \int_{\mathbb{R}_{*}}r^{d-1}dr\int_{\mathbb{R}_{*}}\rho
^{d-1}d\rho\int_{\Delta^{d-1}}d\tau\left(  \theta\right)  \int_{\Delta^{d-1}
}d\tau\left(  \sigma\right)  U_1(r,\rho,\theta,\sigma) \nonumber \\&
= \int_{\mathbb{R}_{*}}r^{d-1}dr\int_{\mathbb{R}_{*}}\rho
^{d-1}d\rho\int_{\Delta^{d-1}}d\tau\left(  \theta\right)  \int_{\Delta^{d-1}
}d\tau\left(  \sigma\right) U_2( r,\rho,\theta,\sigma) \left(  \phi_{R_{1},R_{2}}\left(  r+\rho\right)
-\phi_{R_{1},R_{2}}\left(  r\right) \right)
\end{align}
where 
\begin{equation*}
U_1(r,\rho,\theta,\sigma)= G\left(  r,\rho;\theta,\sigma\right)  F\left(  r,\theta\right) F\left(  \rho,\sigma\right)\frac{\phi_{R_{1},R_{2}}\left(
r+\rho\right)  r\rho}{\left(  r+\rho\right)  }\left\Vert \theta-\sigma
\right\Vert ^{2}
\end{equation*}
and 
\begin{equation*}
U_2(r,\rho,\theta,\sigma)=G\left(  r,\rho;\theta,\sigma\right)  F\left(  r,\theta\right)F\left(  \rho,\sigma\right)  r\left\Vert \theta\right\Vert^{2}.
\end{equation*}

Using \eqref{eq:ffiR12} we can then write \eqref{eq:U12} as%
\begin{equation} 
j_{\frac{R_{1}}{2}}=D_{\frac{R_{1}}{2},R_{2}}+j_{R_{2}}\label{S8E5}%
\end{equation}
where%
\begin{align}
j_{R} &  =\int_{\mathbb{R}_{*}}r^{d-1}dr\int_{\mathbb{R}_{*}}\rho^{d-1}%
d\rho\int_{\Delta^{d-1}}d\tau\left(  \theta\right)  \int_{\Delta^{d-1}}%
d\tau\left(  \sigma\right)  U_2(r,\rho,\theta,\sigma) \left(  \phi_{R}\left(
r\right)  -\phi_{R}\left(  r+\rho\right)  \right) \,, \label{def:jR} \\
D_{\frac{R_{1}}{2},R_{2}} &  =\frac{1}{2}\int_{\mathbb{R}_{*}}r^{d-1}%
dr\int_{\mathbb{R}_{*}}\rho^{d-1}d\rho\int_{\Delta^{d-1}}d\tau\left(
\theta\right)  \int_{\Delta^{d-1}}d\tau\left(  \sigma\right)  U_1(r,\rho,\theta,\sigma)
\,. \label{def:DR1R2}
\end{align}
Note that, since the functions $\phi_{R}$ are decreasing, we have $j_{R}\geq0$. Moreover, given that $\phi_{R_{1},R_{2}}\geq 0$ then 
$D_{\frac{R_{1}}{2},R_{2}}\geq0$. Hence,
\begin{equation}\label{est:DRjR}
0\leq D_{\frac{R_{1}}{2},R_{2}} \leq j_{\frac{R_{1}}{2}}.
\end{equation}

We next show that the integrals $j_{R}$ for any $R> 0$ are
bounded by $C\left\vert J_0\right\vert $ where $J_0$ is the vector flux appearing in the definition of constant flux solution (cf.~Definition \ref{DefConstFlux}) and $C$ is independent on $R$. 
More precisely, we claim that
\begin{equation}\label{est:jRJ0}
0\leq j_{R} \leq C\left\vert J_0\right\vert , \quad \text{for any}\quad R> 0.
\end{equation} 

To prove this we notice that the integrand in the definition of  $j_{R}$ can be non-zero only if $r\leq2R$, $r+\rho>R$.   Furthermore, since
$\phi_{R}(r)-\phi_{R}(r+\rho)\leq \frac{\rho}{R}$ and $\phi_{R}(r)\leq 1 $ we have 
\[
\phi_{R}(r)-\phi_{R}(r+\rho)\leq \min\left\{\frac{\rho}{R},1\right\}. 
\]
Then, using \eqref{def:jR} we obtain
\begin{equation}\label{est:jR}
j_{R}  \leq \iint_{\{r\leq 2R, \, \rho+r \geq \frac{R}{2}\}} r^{d} \rho^{d-1} dr d\rho\int_{\Delta^{d-1}}d\tau\left(  \theta\right)  \int_{\Delta^{d-1}}%
d\tau\left(  \sigma\right)  W(r,\rho;\theta,\sigma) \min\left\{\frac{\rho}{R},1\right\},
\end{equation}
where we set $W(r,\rho;\theta,\sigma)=G\left(  r,\rho;\theta,\sigma\right)  F\left(  r,\theta\right)F\left(  \rho,\sigma\right)$.
Using Definition \ref{DefConstFlux} we have 
\begin{equation}\label{est:AjR}
\int_{\frac{R}{4}}^{4R} \sum_{j=1}^{d} A_j(\xi)d\xi =\frac{15}{4}R \vert J_0\vert \, .
\end{equation}
Hence, thanks to \eqref{S1E8}, we arrive at
\begin{align}\label{est:AjR2}
&\frac{15d}{4}R \vert J_0\vert \nonumber \\
&\quad =\int_{\frac{R}{4}}^{4R} d\xi  \int_{\{ r \leq 4R\} } r^{d} dr \int_{\R_*}\rho^{d-1} d\rho  \int_{\Delta^{d-1}}d\tau\left(  \theta\right)  \int_{\Delta^{d-1}}
d\tau\left(  \sigma\right) \cf{r \leq \xi< r+\rho} W(r,\rho;\theta,\sigma)
\end{align}
where for notational convenience we drop the arguments in $ \cf{r \leq \xi< r+\rho}(\xi, r,\rho)$.
This implies
\begin{align}
&\frac{15 d}{4}R \vert J_0\vert \geq \nonumber \\&
\iint_{\{r\leq 2R, \, \rho+r \geq \frac{R}{2}\}} r^{d} \rho^{d-1} dr d\rho\int_{\Delta^{d-1}}d\tau\left(  \theta\right)  \int_{\Delta^{d-1}} d\tau\left(  \sigma\right)  \int_{\frac{R}{4}}^{4R} d\xi \cf{r \leq \xi< r+\rho} W(r,\rho;\theta,\sigma). \label{est:AjR3}
\end{align}

We now want to prove that in the integrand of (\ref{est:AjR3})
\begin{equation}\label{est:intchi}
\int_{\frac{R}{4}}^{4R} d\xi  \cf{r \leq \xi< r+\rho} \geq \frac{R}{4}  \min\left\{\frac{\rho}{R},1\right\}.
\end{equation}
With this aim we consider separately the cases $\rho \leq r$ and $\rho >r$. 
If $\rho \leq r$, $r\leq 2R$, and $\rho+r \geq \frac{R}{2}$, we have $r\geq \frac R 4$ and $\rho+r \leq 4R$.  Then
\begin{equation*}
\int_{\frac{R}{4}}^{4R} d\xi  \cf{r \leq \xi< r+\rho}=\int_{r}^{\rho+r} d\xi = \rho. 
\end{equation*}

Suppose now that $\rho > r$, $r\leq 2R$,  and $\rho+r \geq \frac{R}{2}$. Then 
\begin{equation*}
\int_{\frac{R}{4}}^{4R} d\xi  \cf{r \leq \xi< r+\rho}=\min\left\{(r+\rho),4R\right\}-\max \left\{r,\frac{R}{4}\right\} = R\psi\left(\frac{r}{R},\frac{\rho}{R}\right)
\end{equation*}
where $\psi\left(y_1,y_2\right)= \min\left\{(y_1+y_2),4\right\}-\max \left\{y_1,\frac{1}{4}\right\} $.  
It turns out that 
$$
\min_{\left\{0\leq y_1\leq 2, \, y_1+y_2\geq \frac 1 2,\, y_1\leq y_2 \right\} } \psi\left(y_1,y_2\right)\geq \frac 1 4 .
$$
This inequality follows considering separately the regions $U_j\cap V_k\cap \{y_1\leq y_2 \}$ for $j=1,2$, $k=1,2$ with
\begin{align*} 
&U_1=\{y_1+y_2\geq 4\} ,\quad U_2=\{y_1+y_2\leq 4 \}, \nonumber \\&
V_1=\{y_1\geq \frac{1}{4}\}, \quad V_2=\{y_1\leq \frac{1}{4}\}.
\end{align*}
Therefore, \eqref{est:intchi} follows. Combining \eqref{est:jR} with \eqref{est:AjR3} and \eqref{est:intchi} we obtain \eqref{est:jRJ0}. 
We now use the inequality \eqref{est:jRJ0} which, together with \eqref{est:DRjR}, yields 
\[
D_{\frac{R_{1}}{2},R_{2}}\leq C\left\vert J\right\vert\, ,
\qquad 0<R_1<R_2\,.
\]
By monotone convergence, this result implies also that
\begin{align}
& D_{0,\infty} :=\frac{1}{2}\int_{\mathbb{R}_{*}}r^{d-1}%
dr\int_{\mathbb{R}_{*}}\rho^{d-1}d\rho\int_{\Delta^{d-1}}d\tau\left(
\theta\right)  \int_{\Delta^{d-1}}d\tau\left(  \sigma\right)  W\left(  r,\rho;\theta,\sigma\right) \frac{r\rho}{r+\rho }\left\Vert \theta-\sigma
\right\Vert ^{2} \nonumber \\
& \quad
\leq C\left\vert J\right\vert \,. \label{def:DR1R2bis}
\end{align}
Our next goal is to use the above convergent integral to prove a dominated convergence argument similar to what was used in the proof of
Theorem \ref{LocContInj}.

First, let us recall that by Proposition  \ref{UppEstCont},
if we set 
\[
Z(R):=\int_{[R,\infty)}r^{\gamma+d-1}
dr\int_{\Delta^{d-1}}d\tau\left(
\theta\right) F\left(  r,\theta\right)\,,
\]
then we may follow the same argument as in the proof of 
Theorem \ref{LocContInj} and find constants such that (\ref{eq:ldenombounds}) holds for all $R>0$.
In particular, we may  define the probability 
measures $\lambda(\theta;R)$ as in \eqref{S7E8} for all $R>0$. Using Lemma \ref{LemmSplit}, we may then conclude via the same argument, now relying on (\ref{def:DR1R2bis}),
that with $\theta_0:= J_0/|J_0|$ we have 
\begin{equation}
\lambda\left(  \theta;R\right)  \rightharpoonup\delta\left(
\theta-\theta_{0}\right)  \ \ \text{as\ \ }R\rightarrow\infty
\,.\label{eq:wcon1bis}
\end{equation}

For the constant flux solution it is also possible to consider the limit $R\to 0$ for the probability distribution
\begin{equation}
\lambda_0\left(  \theta;R\right)  =\frac{\int_{(0,R]}
F\left(  r,\theta\right)  r^{1+d-1}dr}{Z_0(R)
}\,, \qquad
Z_0(R):=\int_{(0,R]}r^{d}
dr\int_{\Delta^{d-1}}d\tau\left(
\theta\right) F\left(  r,\theta\right)\,.
\end{equation}
For this limit, we replaced the power $\gamma$ with $1$
since now, whenever $0<r,\rho\le R$, we have 
$r \rho (r+\rho)^{\gamma-1}\ge 
r\rho (2R)^{\gamma-1}$, and on the other hand
applying Lemma \ref{lem:bound} and Proposition \ref{UppEstCont}  we find constants
$C'_1$ and $C'_2$ such that for all $R>0$
\begin{align} 
 C'_2 \sqrt{|J_0|} R^{\frac{1-\gamma}{2}}\le Z_0(R)\le 
  C'_1 \sqrt{|J_0|} R^{\frac{1-\gamma}{2}}\, .
\end{align}
Following analogous
steps as in the first limit case, it follows that 
\begin{equation}
\lambda_0\left(  \theta;R\right)  \rightharpoonup\delta\left(
\theta-\theta_{0}\right)  \ \ \text{as\ \ }R\rightarrow 0\,,\label{eq:wcon2bis}
\end{equation}
where $\theta_0$ is the same constant vector as in the first case.

It only
remains to show that the Dirac mass occurs not only  for
$R\rightarrow0$ or $R\rightarrow\infty,$ but also for arbitrary values of $R.$
If $0<R_1<R_2$, we can use 
the test function $\psi\left(
r,\theta\right)  =r\phi_{R_{1},R_{2}}\left(  r\right) $
to  obtain $\tilde{j}_{\frac{R_1}{2}}=\tilde{j}_{R_2}$ where
\begin{align*}
 & \tilde{j}_{R}:=\int_{\mathbb{R}_{*}}r^{d-1}dr\int_{\mathbb{R}_{*}}\rho^{d-1}d\rho
\int_{\Delta^{d-1}}d\tau\left(  \theta\right)  \int_{\Delta^{d-1}}d\tau\left(
\sigma\right)  W\left(  r,\rho;\theta,\sigma\right)  \left(  \phi_{R}\left(  r\right)  -\phi
_{R}\left(  r+\rho\right)  \right)  r \,.
\end{align*}
Then, using \eqref{eq:wcon1bis}, \eqref{eq:wcon2bis} together with (\ref{def:DR1R2bis}), we may conclude that $j_{R}-\|\theta_0\|^2 \tilde{j}_{R}\to 0$ both when $R\to 0$ and when $R\to\infty$; note that for any $\theta\in \Delta^{d-1}$
we have $|\norm{\theta}^2-\norm{\theta_0}^2|
\le 2|\norm{\theta}-\norm{\theta_0}|\le 2 \norm{\theta-\theta_0}$.
Thus, using the dissipation formula \eqref{S8E5},  we obtain that $D_{\frac{R_{1}}{2},R_{2}}\to 0$ as $R_1\to 0$, $R_2\to \infty$. This implies that $D_{0,\infty}=0$, and this is possible only if there is a measure $\tilde{H}(r)$ such that 
\[
F\left(  r,\theta\right)  =\tilde{H}\left(  r\right)  \delta\left(  \theta-\theta
_{0}\right).
\]

It is readily seen that, since $F$ solves the coagulation equation
(\ref{weaknoeta}), we can write $\tilde{H}\left(  r\right)  $ as
$\frac{H\left(  r\right)  }{r^{d-1}}$ where $H$ is a constant flux solution
with the kernel $G\left(  r,\rho,\theta_{0},\theta_{0}\right)  $ and the
result follows.
\end{proofof}


 \bigskip

\bigskip

\bigskip

\noindent \textbf{Acknowledgements.}
The authors gratefully acknowledge the support of the Hausdorff Research Institute for Mathematics
(Bonn), through the {\it Junior Trimester Program on Kinetic Theory},
of the CRC 1060 {\it The mathematics of emergent effects} at the University of Bonn funded through the German Science Foundation (DFG), 
 of the {\it Atmospheric Mathematics} (AtMath) collaboration of the Faculty of Science of University of Helsinki, of the ERC Advanced Grant 741487 as well as of  
the Academy of Finland via the {\it Centre of Excellence in Analysis and Dynamics Research} (project No.\ 307333).
The funders had no role in study design, analysis, decision to
publish, or preparation of the manuscript.

\bigskip
\noindent\textbf{Compliance with ethical standards} 
\smallskip

\noindent \textbf{Conflict of interest} The authors declare that they have no conflict of interest.

\bigskip 

 \bigskip
 
 \def\adresse{
\begin{description}

\item[M.~A. Ferreira] 
{Department of Mathematics and Statistics, University of Helsinki,\\ P.O. Box 68, FI-00014 Helsingin yliopisto, Finland \\
E-mail:  \texttt{marina.ferreira@helsinki.fi}}

\item[J. Lukkarinen]{ Department of Mathematics and Statistics, University of Helsinki, \\ P.O. Box 68, FI-00014 Helsingin yliopisto, Finland \\
E-mail: \texttt{jani.lukkarinen@helsinki.fi}}

\item[A. Nota:] {Department of Information Engineering, Computer Science and Mathematics,\\ University of L'Aquila, 67100 L'Aquila, Italy \\
E-mail: \texttt{alessia.nota@univaq.it}}

\item[J.~J.~L. Vel\'azquez] { Institute for Applied Mathematics, University of Bonn, \\ Endenicher Allee 60, D-53115 Bonn, Germany\\
E-mail: \texttt{velazquez@iam.uni-bonn.de}}

\end{description}
}

\adresse
 
 \bigskip

\end{document}